\documentclass[aoas,preprint]{imsart}

\RequirePackage[OT1]{fontenc}
\RequirePackage{amsthm,amsmath}
\RequirePackage{natbib}
\RequirePackage[colorlinks,citecolor=blue,urlcolor=blue]{hyperref}
\usepackage{tikz}
\usetikzlibrary{shapes,backgrounds,arrows,matrix,positioning,chains,fit,calc}
\arxiv{arXiv:0000.0000}

\startlocaldefs
\numberwithin{equation}{section}
\theoremstyle{plain}
\newtheorem{thm}{Theorem}[section]
\newtheorem{condition}{Condition}
\newcommand{\bG}{\boldsymbol\Gamma}
\newcommand{\bg}{\boldsymbol\gamma}

\newcommand{\bZ}{\mathbf Z}

\newcommand{\bH}{\mathbf H}
\newcommand{\bh}{\boldsymbol h}

\newcommand{\bth}{\boldsymbol\theta}
\newcommand{\balpha}{\boldsymbol\alpha}

\newcommand{\bX}{\mathbf X}

\newcommand{\bm}{\boldsymbol m}
\newcommand{\bu}{\boldsymbol u}

\newcommand{\bn}{\boldsymbol n}
\renewenvironment{proof}{{\bf Proof.}}{$\Box$}
\endlocaldefs

\makeatletter
\newcommand\footnoteref[1]{\protected@xdef\@thefnmark{\ref{#1}}\@footnotemark}
\makeatother

\begin{document}

\begin{frontmatter}
\title{Bayesian Propagation of Record Linkage Uncertainty into Population Size Estimation of Human Rights Violations\thanksref{T1}}
\runtitle{Linkage-Averaging for Population Size Estimation}
\thankstext{T1}{Partially supported by NSF grants SES-11-30706 and SES-11-31897.}

\begin{aug}
\author{\fnms{Mauricio} \snm{Sadinle}\ead[label=e1]{msadinle@uw.edu}}
\runauthor{M. Sadinle}

\affiliation{University of Washington}

\address{Department of Biostatistics\\
Department of Statistics\\
Center for Statistics and the Social Sciences\\
University of Washington\\
Box 357232 \\
Seattle, WA 98195\\
\printead{e1}}
\end{aug}

\begin{abstract}
Multiple-systems or capture-recapture estimation are common techniques for population size estimation, particularly in the quantitative study of human rights violations.  These methods rely on multiple samples from the population, along with the information of which individuals appear in which samples.  The goal of record linkage techniques is to identify unique individuals across samples based on the information collected on them.  Linkage decisions are subject to uncertainty when such information contains errors and missingness, and when different individuals have very similar characteristics.  Uncertainty in the linkage should be propagated into the stage of population size estimation.  We propose an approach called \emph{linkage-averaging} to propagate linkage uncertainty, as quantified by some Bayesian record linkage methodologies, into a subsequent stage of population size estimation.  Linkage-averaging is a two-stage approach in which the results from the record linkage stage are fed into the population size estimation stage.  We show that under some conditions the results of this approach correspond to those of a proper Bayesian joint model for both record linkage and population size estimation.  The two-stage nature of linkage-averaging allows us to combine different record linkage models with different capture-recapture models, which facilitates model exploration.  We present a case study from the Salvadoran civil war, where we are interested in estimating the total number of civilian killings using lists of witnesses' reports collected by different organizations.  These lists contain duplicates, typographical and spelling errors, missingness, and other inaccuracies that lead to uncertainty in the linkage.  We show how linkage-averaging can be used for transferring the uncertainty in the linkage of these lists into different models for population size estimation.
\end{abstract}


\end{frontmatter}

\section{Introduction}\label{ch:Intro} 

In the context of armed conflicts, a basic question is \emph{how many human rights violations occurred} in a given time and space.  While a complete enumeration is not typically feasible, it is common to find multiple organizations monitoring and collecting reports on those violations.  Given that witnesses or victims may report an event to different organizations, and different witnesses may report an event to the same organization, the associated record systems often end up containing multiple entries referring to the same violations, even within the same data source.  Those reports may contain different degrees of detail and accuracy, and typically do not contain unique identifiers of the victims, such as national identification numbers.  Therefore, even the more basic question of how many unique human rights violations have been reported cannot be easily answered.  Record linkage techniques are required to detect duplicated reports within each source and to link coreferent reports across data sources.  The result of this linkage stage is often used to derive estimates of the total number of unreported violations using capture-recapture or multiple-systems estimation.  The Human Rights Data Analysis Group --- HRDAG\footnote{Website: https://hrdag.org/} has been a leader and a pioneer in using these methodologies to study human rights violations in several countries \citep[see, e.g.][]{LumPriceBanks13,Priceetal15,PriceBall15}.  Here we revisit a case from El Salvador, where we combine three data sources to explore the question of how many civilians were killed during the Salvadoran civil war (1980--1991) in San Salvador.  

A limitation in this area of application is that population sizes are estimated taking a given linkage of the lists as being the correct one.  Current practice therefore understates the overall uncertainty around the population size as it ignores the uncertainty from the linkage.  We propose a simple procedure called \emph{linkage-averaging} for incorporating the uncertainty from record linkage into subsequent population size estimation using multiple-systems or capture-recapture models.  Linkage-averaging is possible thanks to the advent of Bayesian partitioning approaches that provide proper accounts of the uncertainty in the linkage process \citep[e.g.][]{Matsakis10, Sadinle14, Steortsetal13}.  Linkage-averaging requires two stages.  First, we use a Bayesian partitioning approach to obtain a posterior sample of possible linkages between the lists.  Then, for each of those linkages we obtain a posterior distribution on the population size using a capture-recapture model.  The individual population size posteriors are combined by taking a simple average.  This approach is appealing for being simple and intuitive, and we show that if the capture-recapture model uses only functions of the linkage then linkage-averaging is equivalent to a proper Bayesian approach to joint record linkage and population size estimation.  The two-stage nature of linkage-averaging facilitates model exploration as linkage results can be reused with different capture-recapture models, and it is also well suited for lists with restricted access due to confidentiality constraints given that the information used for the linkage does not have to be transferred to the analyst doing population size estimation. Linkage-averaging has broader applicability since, for example, census coverage evaluation \citep[e.g.][]{EricksenKadaneTukey89,Hogan92,Hogan93,AndersonFienbergbook} and disease prevalence estimation \citep[e.g.][]{LaPorteetal93,MadiganYork97} are also carried out by linking multiple data sources followed by population size estimation.  

We review Bayesian partitioning approaches to record linkage in Section \ref{s:record_linkage}, Bayesian approaches for population size estimation in Section \ref{s:pop_size}, and in Section \ref{s:lapse} we show how to combine them using linkage-averaging.  Finally, in Section \ref{s:ElSalvadorGeneralApplication} we apply this approach to the case study from El Salvador mentioned above.  

\section{Bayesian Partitioning Record Linkage Approaches} \label{s:record_linkage} 

Let $\bX_k$ be the $k$th data source or list, which contains $r_k$ records as its rows, $k=1,\dots,K$.  We define $\bX = (\bX_1,\dots,\bX_K)^T$ as the concatenated list containing all the $r=\sum_k r_k$ records coming from the $K$ different sources.  
The total number of different fields available from the lists is $F$, and if one of these fields is not recorded in a list then it will be missing in $\bX$ for all records coming from that list.  With $n\leq r$ different individuals represented in $\bX$, jointly detecting duplicates within lists and linking records across lists is equivalent to partitioning the rows of $\bX$ into the $n$ groups of coreferent records.  This {\it coreference partition} \citep{Matsakis10,Sadinle14,Steortsetal13} is the parameter of interest in joint duplicate detection and record linkage.

A computationally simple representation of partitions uses arbitrary labelings of the partition's groups.  Let $\bZ=(Z_1,\dots,Z_r)$ be a vector of length $r$ representing a labeling of the records in $\bX$, such that two records receive the same label if and only if they are a match/coreferent. An intuitive way of thinking of $\bZ$ is as an underlying unique identifier that we want to recover. Although the labeling given by $\bZ$ is arbitrary, any equivalent relabeling leads to the same partition of the records, which is what we care about.  Indeed, two records are a match or coreferent if and only if $Z_i=Z_j$.  To fix ideas, the vectors $\bZ=(1, 2, 1, 3, 3)$ and $\bZ=(4, 5, 4, 2, 2)$ are two labelings of the same partition of five elements, since in both $Z_1=Z_3\neq Z_4=Z_5$, and $Z_2$ gets its own value.  

From a Bayesian point of view, one obtains a posterior distribution on $\bZ$ given $\bX$, and the variability captured by this posterior should ideally reflect the uncertainty in the record linkage and duplicate detection procedure.  There exist two types of approaches to obtain such a posterior on $\bZ$: direct modeling approaches and comparison-based approaches.

\subsection{Direct-Modeling Approaches}

A number of Bayesian approaches to both duplicate detection and record linkage have been proposed where one directly models the information contained in the lists/datafiles \citep{Matsakis10, TancrediLiseo11, Fortinietal02, Gutmanetal13, Steortsetal13},  that is, one proposes a model $P(\bX\mid \bZ)$ for the information observed in the lists, and a posterior on $\bZ$ is derived as $p(\bZ\mid \bX)\propto p(\bZ)P(\bX\mid \bZ)$, with the help of a prior on partitions $p(\bZ)$.  To write down $P(\bX\mid \bZ)$ one needs crafting specific models for each type of field in the lists.  The models of \cite{Matsakis10, TancrediLiseo11, Steortsetal13, Steorts15} share the characteristic that given a value of $\bZ$, the clusters of coreferent records are modeled as distortions of some latent record containing the true information of a latent individual.  These approaches currently mostly handle categorical fields, with the exceptions of \cite{Steorts15} who proposed an empirical Bayes approach to model names, and \cite{LiseoTancredi11} who handle continuous fields under normality.  In practice, however, fields that are complicated to model, such as strings, addresses, phone numbers, or dates, are also important to detect coreferent records.  These type of fields are often subject to typographical and other types of errors, which makes it important to take into account partial agreements between their values.  Existing direct modeling approaches also currently do not handle missing data, although this extension should be easy to implement.

\subsection{Comparison-Based Approaches}\label{ss:comparison-based}

A number of approaches to record linkage and duplicate detection rely on the often reasonable assumption that two records referring to the same entity should be very similar.  If this is not the case in a given application scenario then the task of finding coreferent records might be hopeless.  Comparison vectors are computed for pairs of records from $\bX$ to summarize evidence of whether they refer to the same entity.  For record pair $(i,j)$ we compare each field $f=1,\dots,F$ by computing some similarity measure $\mathcal{S}_f(i,j)$, which depends on the type of information contained by each field.  For unordered categorical fields like sex or race, $\mathcal{S}_f$ can simply be a binary comparison checking whether the records agree in that field.  For more structured fields, $\mathcal{S}_f$ should be able to capture partial agreements.  For example, in the case of strings such as names or addresses, $\mathcal{S}_f$ should correspond to a string metric, such as the Levenshtein edit distance, the Jaro--Winkler score, or any other \citep[see][]{Bilenkoetal03,Elmagarmidetal07}.  Some of these comparisons will be missing, since if  field $f$ is missing for a record $i$, then $\mathcal{S}_f(i,j)$ will be missing regardless of whether field $f$ is observed for record $j$.

In principle, we could define the comparison vectors using the original similarity values $\mathcal{S}_f(i,j)$, $f=1,\dots,F$, but 
the direct modeling of the $\mathcal{S}_f(i,j)$'s requires customized models per type of comparison, because the outputs of these similarity measures lie in different spaces, depending on the type of field being compared.  Instead, \cite{Sadinle14} followed \cite{Winkler90Strings} in dividing the range of each similarity measure $\mathcal{S}_f$ into $L_f+1$ intervals $I_{f0}, I_{f1},\dots, I_{fL_f}$, that represent different levels of disagreement.  In this construction we associate the interval $I_{f0}$ with the highest level of agreement, including no disagreement, and the last interval, $I_{fL_f}$, with the highest level of disagreement, which depending on the field may represent complete or strong disagreement. For records $i$ and $j$, and field $f$, we define
\begin{equation*}
\gamma^{f}_{ij} = l, \hbox{ if \  } \mathcal{S}_f (i,j) \in I_{fl}.
\end{equation*}
As the value of $\gamma^{f}_{ij}$ increases, so does the disagreement between records $i$ and $j$ with respect to field $f$.  The possible values of $\gamma^{f}_{ij}$ simply represent the categories of an ordinal variable.  We then define the comparison vector $\bg_{ij}=(\gamma_{ij}^1,\dots,\gamma_{ij}^f,\dots,\gamma_{ij}^F)$ for records $i$ and $j$.  Building comparison data as ordinal categorical variables facilitates modeling since we can use a generic model for any type of comparison, as long as its values are categorized in a meaningful way.

A number of traditional record linkage and duplicate detection approaches use pairwise comparisons $\bg_{ij}$, but they output independent pairwise decisions on the matching/coreference status of pairs of records \citep{FellegiSunter69,Winkler88, Jaro89,LarsenRubin01}, which then need to be reconciled in some ad-hoc manner as they may not be compatible with one another.  \cite{Sadinle14} modified comparison-based approaches to  directly target $\bZ$ rather than pairwise matching decisions.  Letting $\bG(\bX)$ denote the comparison data for all record pairs, the approach of \cite{Sadinle14} corresponds to a model  $P(\bG(\bX)\mid \bZ)$ which along with a prior $p(\bZ)$ allows us to obtain a posterior $p(\bZ\mid \bG(\bX))$.

The model for the comparison data  $\bG(\bX)$ presented by \cite{Sadinle14} assumes that $\bg_{ij}$ is a realization of a random vector $\bG_{ij}$ such that:
\begin{eqnarray*}
\bG_{ij}\mid Z_i=Z_j ~ \overset{iid}{\sim} ~ G_1; ~~
\bG_{ij}\mid Z_i\neq Z_j ~ \overset{iid}{\sim} ~ G_0.
\end{eqnarray*}
In this model, $G_1$ and $G_0$ represent the distributions of the comparison vectors among coreferent and non-coreferent pairs, respectively. 

\cite{Sadinle14} parameterized $G_1$ as
\begin{equation}\label{eq:DDRL_P1missing}
P_1(\bg^{obs}_{ij}\mid \Phi_1) = \prod_{f=1}^{F}
\Bigg[\prod_{l=0}^{L_f-1} (m_{fl})^{I(\gamma^{f}_{ij}=l)}(1-m_{fl})^{I(\gamma^{f}_{ij}>l)}\Bigg]^{I_{obs}(\gamma_{ij}^f)},
\end{equation}
which is obtained under conditional independence of the comparison fields, and ignorability of the missingness in the comparison vectors.  
 $I_{obs}(\gamma_{ij}^f)$ indicates whether $\gamma_{ij}^f$ is observed, $\Phi_1=(\bm_1,\dots,\bm_F)$, with $\bm_f=(m_{f0},\dots,m_{f,L_f-1})$, where $m_{f0}=P_1(\Gamma^f_{ij}=0)$, and $m_{fl}=P_1(\Gamma^f_{ij}=l\mid \Gamma^f_{ij}>l-1)$ for $0<l<L_f$.  A similar expression can be obtained for $P_0(\bg^{obs}_{ij}\mid \Phi_0)$ in terms of parameters $\Phi_0=(\bu_1,\dots,\bu_F)$.  

The parameterization in terms of the sequential conditional probabilities $m_{fl}$ facilitates prior specification.  The parameter $m_{fl}=P_1(\Gamma^f_{ij}=l\mid \Gamma^f_{ij}>l-1)$ represents the probability of observing disagreement level $l$ in the comparison $f$, among two coreferent records with disagreement larger than level $l-1$.  Unless we expect field $f$ in one of these two datafiles to be highly unreliable, we would a priori expect each $m_{fl}$ to be fairly close to 1.  For example, for $l=0$ this is simply $m_{f0}=P_1(\Gamma^f_{ij}=0)$, which represents the marginal probability of disagreement level zero, which encodes full or a high degree of agreement, and so $m_{f0}$ should be high if the field $f$ in these two datafiles does not contain too many errors.  For $l=1$, we have $m_{f1}=P_1(\Gamma^f_{ij}=1\mid \Gamma^f_{ij}>0)$, which represents the probability of observing disagreement level one in the comparison $f$, among coreferent record pairs with disagreement larger than what is captured by the level zero.  If the number of disagreement levels is greater than two, we can think of level one of disagreement as a type of mild disagreement, meaning that if we expect the amount of error to be relatively small, then $m_{f1}$ should be concentrated around values close to one.  As we consider other parameters $m_{fl}$ for levels $l>2$, it is easy to see that they should also be close to one, if field $f$ does not contain too many errors.  In general, we can therefore think of using the priors $m_{fl}\sim\text{Uniform}[\lambda_{fl},1]$, for some prior truncation points $0<\lambda_{fl}<1$, such that the less accurate we believe field $f$ is, the lower the value for $\lambda_{fl}$.  More generally, we could take truncated beta priors, but here we focus on specifying our prior beliefs through these truncation points $\lambda_{fl}$.  

It is more difficult to incorporate prior information on the probabilities $u_{fl}=P_0(\Gamma^f_{ij}=l\mid \Gamma^f_{ij}>l-1)$, since the distribution of the disagreement levels among non-coreferent record pairs may be quite different depending on the characteristics of the fields.  For example, a categorical field with a highly frequent category will lead to a high probability of $\Gamma^f_{ij}=0$ even for non-coreferent record pairs, but a field like phone number or address will lead to small probabilities of agreement among non-coreferent record pairs.  For simplicity we therefore take each $u_{fl}\sim\text{Uniform}(0,1)$.

The approach of \cite{Sadinle14} heavily relies on being able to reduce the set of candidate coreferent record pairs on which vectors of comparisons are computed.  By using simple rules that can efficiently identify non-coreferent pairs we seek to avoid comparing all the $r(r-1)/2$ record pairs when $r$ is large. 
For example, if the datafiles contain a categorical field deemed to be error-free, one can simply take records disagreeing on that field as being non-coreferent.  This simple approach is known as {\it blocking}.  Unfortunately, in many applications all fields may be subject to error, and therefore we need to devise other ways of filtering non-coreferent records.  An alternative is to exploit prior knowledge on the kinds of errors that would be unlikely for a certain field, thereby declaring as non-coreferent any record pair that disagrees more than a predefined amount in that field.  
There also exist other more sophisticated techniques to detect sets of non-coreferent pairs, which are extensively surveyed by \cite{Christen12}.  

After this initial filtering step, the set $\mathcal{P}$ comprises the remaining candidate coreferent record pairs, on which we compute comparison vectors.  Using these comparison vectors we define additional rules to fix record pairs as non-coreferent.
For instance, strong disagreements in both given and family names, or in other combination of fields may be a robust indication of the pair being non-coreferent. The final set of candidate coreferent pairs is $\mathcal{C}\subseteq\mathcal{P}$.  

The possible coreference partitions are finally constrained to the set $\mathcal{Z}=\{\bZ: Z_{i}\neq Z_j,  \ \forall \ (i,j)\notin \mathcal{C}\}$, that is, any partition that puts together record pairs already declared as non-coreferent is unfeasible. The approach of \cite{Sadinle14} relies on $\mathcal{Z}$ being much smaller than the set of all possible partitions, which is why we heavily rely on being able to obtain a small set of candidate pairs $\mathcal{C}$.  The comparison vectors of the pairs in $\mathcal{P}\setminus\mathcal{C}$ are used in the model but fixed as non-coreferent pairs, that is, they never get assigned the same label in $\bZ$.  The prior distribution on $\bZ$ used by  \cite{Sadinle14} was derived from a uniform distribution on partitions constrained to partitions consistent with $\mathcal{Z}$. A simple way to obtain the flat prior on partitions from a prior for $\bZ$ is by assigning equal probability to each of the $r!/(r-n)!$ labelings of a partition with $n$ groups, which leads to the prior on partition labelings $p(\bZ)\propto [(r-n(\bZ))!/r!]I(\bZ\in \mathcal{Z})$, where $n(\bZ)$ measures the number of different labels in labeling $\bZ$.

Finally, \cite{Sadinle14} developed a Gibbs sampler to obtain draws from the posterior distribution of $\bZ$.

\subsection{A Practical Comparison of Bayesian Partitioning Record Linkage Approaches}

Both direct-modeling and comparison-based approaches have advantages and disadvantages when compared to one another.  One can argue that direct-modeling approaches are more principled, as they directly model the records in the datafiles/lists.  Instead, comparison-based approaches merely model comparisons between pairs of records.  This advantage of direct-modeling approaches can also be seen as a disadvantage, as the lists $\bX$ may contain some combination of fields that are difficult to directly model like family and given names, dates, addresses, phone numbers, etc.  Writing $P(\bX\mid \bZ)$  requires proposing models for such fields, which requires modeling how such information gets corrupted.  Comparison-based approaches have an advantage here, because any type of field can be used to construct the comparison data, as long as the comparisons are meaningful for the fields at hand.  Therefore, models $P(\bG(\bX)\mid \bZ)$ will often be much simpler than models $P(\bX\mid \bZ)$.  

In this article we will use the comparison-based approach of \cite{Sadinle14}, which is better suited to the data from El Salvador.  Direct-modeling approaches currently do not  handle missing data and need computational speed-ups. For example, the approach of \cite{Steorts15} as implemented in the R package \texttt{blink} takes 8.4 hours to compute 30,000 MCMC iterations with a file of size 500 included in the \texttt{blink} package, and requires around 10,000 iterations to reach convergence.  By contrast, the approach of \cite{Sadinle14} can take advantage of fixing obvious non-coreferent record pairs as non-coreferent, which leads to a much faster Gibbs sampler.  With the file of size 500 included in the \texttt{blink} package, after fixing record pairs with high Levenshtein distance in first or last name as non-coreferent, we obtain 15,052 candidate coreferent pairs.  30,000 iterations of the Gibbs sampler of \cite{Sadinle14} run in one hour, but convergence is achieved in less than 10 iterations.  This comparison was done on a laptop with a processor Intel Core i7-4900MQ.

Regardless of what approach one uses, the critical requirement needed in this article is that the record linkage approach provides a set of draws $\bZ^{(1)},\bZ^{(2)},\dots,\bZ^{(d)}$ from a posterior $p(\bZ\mid \bG(\bX))$, $\bG(\bX)$ being the comparison data in the case of comparison-based approaches, or from $p(\bZ\mid \bX)$ in the case of direct-modeling approaches.

\section{Population Size Estimation}\label{s:pop_size}

To estimate the total number of units or individuals in a closed population, a number of techniques rely on the availability of $K$ incomplete lists/samples drawn from the population.  The name \emph{capture-recapture} comes from applications in population ecology where the goal is to estimate animal abundance.  In that context the technique consists in drawing $K$ samples from the population in a sequential manner while keeping track of the individuals' inclusion patterns, that is, which individuals have been included in which samples \citep[see, e.g.][]{Pollock00}.  In the context of estimating the size of human populations, the $K$ samples often come from  record systems which are not necessarily collected in a sequential manner, but are represented by datafiles or lists containing (partially) identifying information on the individuals.  In that context the term \emph{multiple-systems estimation} is often preferred \citep[see, e.g.][]{BirdKing18}.  The discussion in this article applies to capture-recapture or multiple-systems estimation models with sufficient statistics that depend only on the inclusion patterns of the different individuals \citep[e.g.,][]{Fienberg72,Bishopetal75,Castledine81,GeorgeRobert92,MadiganYork97,FienbergJohnsonJunker99,Manrique16}.

Let an inclusion pattern be represented by a vector $\bh=(h_1,\dots,h_K)$ in $\{0,1\}^K$, where $h_k=1$ indicates inclusion in the record-system $k$.  Let $n_{\bh}$ represent the number of individuals with inclusion pattern $\bh$.  The  inclusion patterns' frequencies can be organized in a contingency table $\bn^*=\{n_{\bh} \}_{\bh\in\{0,1\}^K}$.  Notice that we do not observe the number of individuals missed by all record-systems, that is, $n_{00\dots 0}$ is missing, and so we let $\bn=\{n_{\bh}\}_{\bh\in\{0,1\}^K\setminus\{0\}^K}$ represent the observed counts.   For example, with three record-systems we denote the observed frequencies of the different inclusion patterns as $\bn = \{n_{111}, n_{011}, n_{101}, n_{001}, n_{110}, n_{010}, n_{100}\}$, where, for example, $n_{101}$ represents the number of individuals included in record-systems one and three but not in record-system two.

For a given individual we can think of their inclusion pattern $\bh$ as a realization of a $K$-variate binary vector such that $P(\bh\mid \bth)=\theta_{\bh}$, with the vector $\bth=\{\theta_{\bh}\}_{\bh\in\{0,1\}^K}$ providing the probability of each inclusion pattern.  Let $\bth(m)$ denote the capture probabilities as dictated by a model $m$.  Given that there are $N=\sum_{\bh\in\{0,1\}^K}n_{\bh}$ individuals in the population, under the assumption that their inclusion patterns are independent and identically distributed, we have that the joint distribution of the contingency table $\bn^*$ is multinomial with probability mass function
\begin{equation}\label{eq:Pn_nthm}
P(\bn^*\mid N,\bth(m),m) = N!\prod_{\bh\in\{0,1\}^K}\frac{\theta_{\bh}(m)^{n_{\bh}}}{n_{\bh}!}.
\end{equation}
Notice that since for given $N$ and $\bn$ we can obtain $n_{00\dots 0}=\linebreak N-\sum_{\bh\in\{0,1\}^K\setminus\{0\}^K}\ n_{\bh}$, we can write $P(\bn\mid N,\bth(m),m)=P(\bn^*\mid N,\bth(m),m)$.

Given a model $m$ and a prior on the population size $p(N)$, we are interested in obtaining a posterior distribution 
\begin{equation}\label{eq:PN_nm}
P(N\mid \bn,m) = \frac{P(\bn\mid N,m)p(N)}{\sum_N P(\bn\mid N,m)p(N)},
\end{equation}
where
\begin{align}\label{eq:Pn_Nm_generic}
P(\bn\mid N,m) & = \int_{\bth(m)} P(\bn\mid N,\bth(m),m)p(\bth(m)\mid m)d\bth(m),
\end{align}
for a prior on the model parameters $p(\bth(m)\mid m)$, assuming that $N$ and $\bth(m)$ are independent a priori.  

As mentioned before, a number of approaches for population size estimation fit into this description, but for simplicity we only describe the approach based on decomposable graphical models of \cite{MadiganYork97} and the approach based on mixture models of \cite{Manrique16}.

\subsection{Approaches Based on Graphical Models} \label{ss:graphicalmodels}

It is especially convenient to work with models and priors that allow a closed form for $P(\bn\mid N,m)$ in \eqref{eq:Pn_Nm_generic}. 
 \cite{MadiganYork97} present one class of graphical models that have this characteristic.  Probabilistic graphical models \citep[see, e.g.,][]{Lauritzen96, Edwards00} provide a way of encoding the set of conditional independencies of a multivariate distribution into a graph.  In a graphical model, each random variable is represented by a node in a graph, and two nodes are joined by an edge if the variables are conditionally dependent given a set of other variables.  In the context of this article a graphical model captures conditional independencies between the binary variables that indicate inclusion of the individuals into the lists $\bX_1,\dots,\bX_K$.  A graphical model $m$ will depend on a set of parameters $\bth(m)$ that satisfy certain constraints dictated by the independencies in the graph.  \cite{MadiganYork97} further restrict their attention to the class of \emph{decomposable} graphical models, which are characterized by their independence graph being chordal (triangulated).  The first two columns of Table \ref{t:3var_graphs} present all non-saturated graphical models for three samples/lists, in which case all happen to be decomposable.  \cite{DawidLauritzen93} introduced the \emph{hyper-Dirichlet} distributions, which can be used as priors for the parameters $\bth(m)$ in such models, and lead to closed formulae for $P(\bn\mid N,m)$.  For the sake of this article, it is enough to say that the parameters of a hyper-Dirichlet prior can be specified from thinking on a table $\balpha=\{\alpha_{\bh}\}_{\bh\in\{0,1\}^K}$ of ``prior counts'' of the same size as $\bn^*$.  In this document we will think of all the entries of $\balpha$ being a constant $\alpha$, in particular $\alpha=1$.  Given a hyper-Dirichlet prior for the model parameters $\bth(m)$, and if $N$ and $\bth(m)$ are independent a priori, \cite{MadiganYork97} show that
\begin{align}\label{eq:Pn_Nm}
P(\bn\mid N,m) & = \int_{\bth(m)} P(\bn\mid N,\bth(m),m)p(\bth(m)\mid m)d\bth(m)\nonumber\\
& = \frac{N!}{\prod_{\bh\in\{0,1\}^K}n_{\bh}!}\frac{\Psi_m(\balpha+\bn^*)}{\Psi_m(\balpha)},
\end{align}
where 
\begin{equation}\label{eq:Psi}
\Psi_m(\balpha)=\frac{\prod_{l=1}^{L}\prod_{\bh_{C_l}}\Gamma(\alpha_{\bh_{C_l}})}{\Gamma(\sum_{\bh\in\{0,1\}^K} \alpha_{\bh})^Q\prod_{l=2}^{L}\prod_{\bh_{S_l}}\Gamma(\alpha_{\bh_{S_l}})}.
\end{equation}
In this expression $\{C_l\}_{l=1}^L$ represents the set of (maximal) cliques,  $\{S_l\}_{l=2}^L$ the set of separators (including multiplicities), and $Q$ the number of connected components of the independence graph of model $m$.  For a given subset of nodes $A$, $\bh_{A}$ represents an inclusion pattern constrained to the variables in $A$.  Finally, $\alpha_{\bh_{A}}=\sum_{\bh':\bh'_{A}=\bh_{A}}\alpha_{\bh'}$.  (Notice that Equation \eqref{eq:Psi} appears with $Q=1$ in \cite{MadiganYork97}, but if we do not take the number of connected components into account then $P(\bn\mid N,m)$ does not add up to 1).

With the methodology of \cite{MadiganYork97} we can also take into account the uncertainty on the model for population size estimation as
\begin{equation}\label{eq:PN_n}
P(N\mid \bn)= \frac{p(N)\sum_{m}P(\bn\mid N,m)p(m)}{\sum_N p(N)\sum_{m}P(\bn\mid N,m)p(m)},
\end{equation}
for a prior $p(m)$ on a finite number of models.  In this article we take $p(m)$ to be uniform over the class of models.  For three lists, there are seven non-saturated decomposable graphical models, and so $p(m)=1/7$.

\subsection{Approaches Based on Mixture Models} 

An alternative model $m$ for the probabilities of the inclusion patterns $P(\bh\mid\bth(m))=\theta_{\bh}(m)$ is obtained by assuming the existence of strata $s=1,\dots,S$, such that inside each of them the inclusion indicators are independent of each other, that is, $P(\bh\mid s, \bth_s)=\prod_{k=1}^K \theta_{sk}^{h_k}(1-\theta_{sk})^{1-h_k}$, where $P(h_k=1\mid s, \bth_s)=\theta_{sk}$ is the probability of an individual being included in list $k$ given that it belongs to stratum $s$.  Each stratum has a probability $\pi_s$, $\sum_{s=1}^S\pi_s=1$.  The probability of the inclusion patterns under this mixture model approach is therefore $\theta_{\bh}(m)=\sum_{s=1}^S\pi_s \prod_{k=1}^K \theta_{sk}^{h_k}(1-\theta_{sk})^{1-h_k}$, which can then be plugged into \eqref{eq:Pn_nthm}.

\cite{Manrique16} used the priors $\theta_{sk}\sim \text{Beta}(1,1)$, and expressed each $\pi_s=V_s\prod_{t<s}(1-V_t)$ where each $V_t\sim\text{Beta}(1,\alpha)$, $t=1,\dots,S-1$, $V_S=1$, and $\alpha\sim \text{Gamma}(.25,.25)$.  This construction is known as a finite-dimensional stick-breaking prior \citep{IshwaranJames01} and it encourages most of the mass to be concentrated in the initial $\pi_s$'s, which consequently makes the choice of $S$ irrelevant as long as it is relatively large.  These priors would in principle allows us to integrate $\bth(m)$ as in \eqref{eq:Pn_Nm_generic}, and then obtain \eqref{eq:PN_nm}, but in this case these integrals are not easily computable, which is why \cite{Manrique16} developed an MCMC algorithm to obtain posterior samples from \eqref{eq:PN_nm} under this mixture model approach.  For further details on this approach see \cite{Manrique16}.

\section{Linkage-Averaged Population Size Estimation} \label{s:lapse} 

\subsection{Derivation of Inclusion Patterns}

We start by explaining how to compute the incomplete contingency table $\bn$ from a given coreference partition labeling $\bZ$.  Let $n$ be the number of different labels in $\bZ$, that is, $n$ represents the number of different individuals that are included in the $K$ datafiles/lists according to the coreference partition represented by $\bZ$.  Without loss of generality we can think of the labels in $\bZ$ to be $1,\dots,n$.  If this is not the case we can simply obtain an equivalent labeling that uses those labels. Now, for each different label $z=1,\dots,n$, let 
\begin{equation*}
h_{zk}=
\left\{
  \begin{array}{ll}
    1, & \hbox{ if there exists a record $i\in\bX_k$ such that $Z_i=z$;} \\
    0, & \hbox{ otherwise.}
  \end{array}
\right.
\end{equation*}
The vector $\bH_k=(h_{1k},\dots,h_{nk})$ contains the indicators of whether each of the $n$ individuals is included in the $k$th datafile.  The contingency table $\bn$ is simply obtained as a cross-classification of these $K$ inclusion vectors.  We write $\bn(\bZ)$ to emphasize that the contingency table $\bn$ is a function of a coreference partition represented by $\bZ$.  

\subsection{Linkage-Averaged Population Size Estimation}\label{ss:lapse}

The output that we use from the record linkage and duplicate detection stage is a posterior sample $\bZ^{(1)},\dots,\bZ^{(d)}$ from a posterior $p(\bZ\mid \bX)$ or $p(\bZ\mid \bG(\bX))$, as exemplified in Figure \ref{f:example}.  

\begin{figure*}[h]
\centering
		\centerline{\includegraphics[width=.9\linewidth]{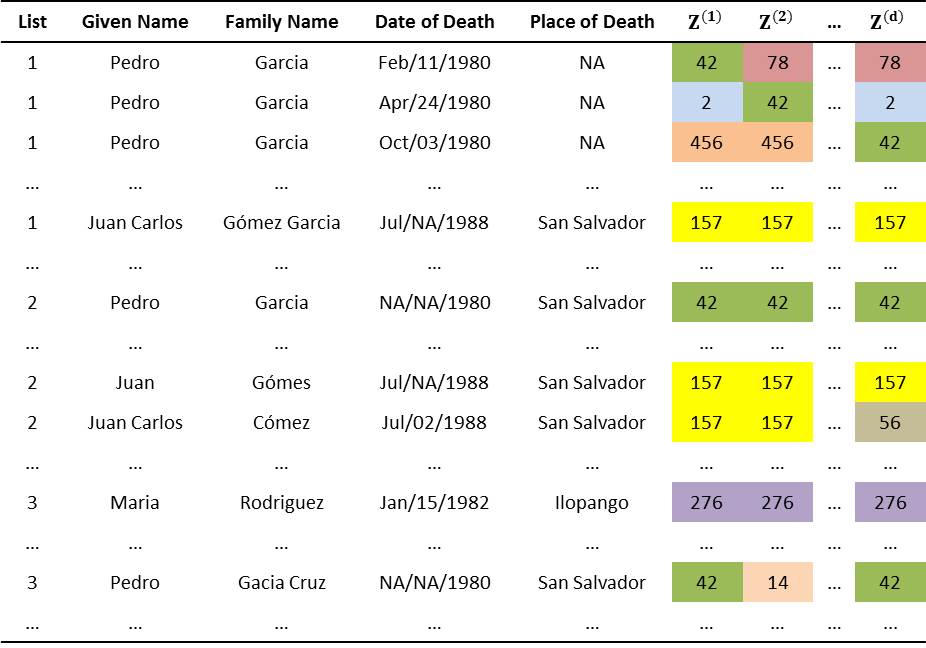}}
  \begin{minipage}[b]{1\textwidth}
  \caption{Illustration of posterior draws $\bZ^{(1)},\dots,\bZ^{(d)}$ obtained from a Bayesian partitioning methodology for record linkage and duplicate detection. The draws $\bZ^{(1)},\dots,\bZ^{(d)}$ can be informally interpreted as ``plausible unique identifiers'' for the individuals in the lists.}
\label{f:example}\end{minipage} 
\end{figure*}

For each of these draws, we can compute the implied contingency tables containing the frequencies of the inclusion patterns $\bn(\bZ^{(1)}),\dots,\bn(\bZ^{(d)})$.  For each of these contingency tables, we can obtain a posterior distribution on the population size using one of the capture-recapture models in Section \ref{s:pop_size}, that is, we can obtain $p(N\mid \bn(\bZ^{(1)})),\dots,p(N\mid \bn(\bZ^{(d)}))$, or a Monte Carlo approximation of these.  The linkage-averaged  posterior of $N$, $p_{\textsc{la}}(N)$, defined formally in the next Section, is approximated as
\begin{equation}\label{eq:LA1}
p_{\textsc{la}}(N) \approx \frac{1}{d}\sum_{t=1}^d p(N\mid \bn(\bZ^{(t)})),
\end{equation}
when each $p(N\mid \bn(\bZ^{(t)}))$ is available in closed form, as with the methodology of \cite{MadiganYork97}.  When this is not the case, as with the approach of \cite{Manrique16}, we use a random sample $N^{(1,t)}, \dots N^{(b,t)} \sim p(N\mid \bn(\bZ^{(t)}))$, for each $t=1,\dots,d$, and use the approximation
\begin{equation}\label{eq:LA2}
p_{\textsc{la}}(N) \approx \frac{1}{db}\sum_{t=1}^d\sum_{v=1}^b  I(N=N^{(v,t)}).
\end{equation}
The formal justification for this linkage-averaged posterior is given next.

\subsection{Bayesian Justification of Linkage-Averaging}

Our strategy for incorporating linkage uncertainty into population size estimation can be derived from a proper Bayesian analysis under two reasonable conditions.

\begin{condition}\label{c1} Our beliefs on $\bZ$ are represented by the posterior distribution $p_{\textsc{l}}(\bZ\mid \bX)$, coming from a model for record linkage and duplicate detection, composed by a likelihood function $\mathcal{L}_{\textsc{l}}(\bZ \mid \bX)$ and a prior $p(\bZ)$.  
\end{condition}

For our discussion, the linkage model can be one of the ones presented in Section \ref{s:record_linkage}, but we only require it to provide a proper posterior distribution on coreference partitions.  For simplicity we use the notation $p_{\textsc{l}}(\bZ\mid \bX)$ and $\mathcal{L}_{\textsc{l}}(\bZ \mid \bX)$ to represent models that either directly model the fields in the lists $\bX$ or that use comparison data, although for the latter the notation $p_{\textsc{l}}(\bZ\mid \bG(\bX))$ and $\mathcal{L}_{\textsc{l}}(\bZ\mid \bG(\bX))$ would be more appropriate, with $\bG(\bX)$ representing the comparison data built from the records in $\bX$.  

\begin{condition}\label{c2} If we knew the true value of $\bZ$, our beliefs on the population size $N$ would be represented by the posterior distribution $p_{\textsc{c}}(N\mid \bn(\bZ))$, obtained from a capture-recapture model composed by a likelihood function \linebreak $\mathcal{L}_{\textsc{c}}(N\mid \bn(\bZ))$ and a prior $p(N)$.  
\end{condition}

This condition simply indicates how we would obtain inferences on $N$ if we knew which records were coreferent.  Note that the likelihood function $\mathcal{L}_{\textsc{c}}(N\mid \bn(\bZ))$ should come from a capture-recapture model that has the frequencies of the inclusion patterns $\bn(\bZ)$, or a function of them, as sufficient statistics, such as those discussed in Section \ref{s:pop_size}.  In particular, notice that the capture-recapture model could involve only a subset of the $K$ datafiles being linked, that is, it could depend on inclusion patterns only for a subset of the $K$ datafiles.  This scenario could arise in cases where some of the datafiles being linked arise from collection mechanisms that make the assumptions of the capture-recapture model seem implausible, such as lists that target members of the population with a distinctive trait and therefore lead to zero probability of inclusion for individuals without the trait.  

Given the setup of Conditions \ref{c1} and \ref{c2}, it seems natural to compute
\begin{equation*}
p_{\textsc{la}}(N)\equiv E_{\bZ\mid \bX}[p_{\textsc{c}}(N\mid \bn(\bZ))]=\sum_{\bZ} p_{\textsc{c}}(N\mid \bn(\bZ)) p_{\textsc{l}}(\bZ\mid \bX),
\end{equation*}
as a way of propagating the linkage uncertainty into population size estimation.  We refer to $p_{\textsc{la}}(N)$ as the linkage-averaged population size posterior.  Here, $p_{\textsc{la}}(N)$ corresponds to the expected posterior distribution of the population size, averaging with respect to the posterior distribution of the coreference partition.  This procedure is intuitively appealing, $p_{\textsc{la}}(N)$ has a clear interpretation, and we now show that $p_{\textsc{la}}(N)$ also corresponds to a proper posterior distribution.  

In principle, if we want to draw inferences jointly on $N$ and $\bZ$ given $\bX$, we need to specify a joint prior $p(N,\bZ)$.  From Condition \ref{c2}, we have that the distribution $p_{\textsc{c}}(N\mid \bn(\bZ))$ would contain our belief on the population size if $\bZ$ was known.  Similarly, from Condition \ref{c1} we have that the prior $p(\bZ)$ contains our prior beliefs on $\bZ$.  Therefore, Conditions \ref{c1} and \ref{c2} imply the joint prior $p(N,\bZ)=p_{\textsc{c}}(N\mid \bn(\bZ))p(\bZ)$.  

\begin{thm}[Bayesian propriety of linkage-averaged population size posterior]\label{validity}
$p_{\textsc{la}}(N)$ is the marginal posterior distribution of $N$ under the likelihood of the linkage model $\mathcal{L}_{\textsc{l}}(\bZ\mid \bX)$ and the joint prior $p_{\textsc{c}}(N\mid\bn(\bZ))p(\bZ)$.
\end{thm}
\begin{proof}
The joint posterior of $N$ and $\bZ$ is 
$p(N,\bZ\mid \bX) \propto \linebreak \mathcal{L}_{\textsc{l}}(\bZ\mid \bX) p_{\textsc{c}}(N\mid \bn(\bZ)) p(\bZ)$, where the inverse of the proportionality constant is $\sum_{\bZ}\sum_{N}\mathcal{L}_{\textsc{l}}(\bZ\mid \bX)p_{\textsc{c}}(N\mid \bn(\bZ))p(\bZ) =\sum_{\bZ}\mathcal{L}_{\textsc{l}}(\bZ\mid \bX)p(\bZ)$, since $\sum_{N}p_{\textsc{c}}(N\mid \bn(\bZ))=1$.  Given that 
$p_{\textsc{l}}(\bZ\mid \bX)\propto\mathcal{L}_{\textsc{l}}(\bZ\mid \bX) p(\bZ)$ with the inverse of the proportionality constant being $\sum_{\bZ}\mathcal{L}_{\textsc{l}}(\bZ\mid \bX) p(\bZ)$, we can therefore write $p(N,\bZ\mid \bX)=p_{\textsc{c}}(N\mid \bn(\bZ)) p_{\textsc{l}}(\bZ\mid \bX)$.  Then,
$p(N\mid \bX)  = \sum_{\bZ} p(N,\bZ\mid \bX) = \sum_{\bZ} p_{\textsc{c}}(N\mid \bn(\bZ)) p_{\textsc{l}}(\bZ\mid \bX) = p_{\textsc{la}}(N)$.
\end{proof}

Furthermore, the total variability represented by $p_{\textsc{la}}(N)$ can be decomposed as
\begin{align}\label{eq:var_decomp}
Var(N\mid \bX)   = \ & 
    Var_{\bZ\mid \bX}[E(N\mid \bZ)] 
    +\ E_{\bZ\mid \bX}[Var(N\mid \bZ)],
\end{align}
where the first term on the right hand side can be seen as the contribution of the linkage uncertainty on the population size variability, and the second term summarizes the variability that is intrinsic to Bayesian approaches for estimating $N$. 

In practice, we generally will  have to approximate  $p_{\textsc{la}}(N)$ and the variance components in \eqref{eq:var_decomp} using posterior draws from $p_{\textsc{l}}(\bZ\mid \bX)$, as explained in Section \ref{ss:lapse} and Equations \eqref{eq:LA1} and \eqref{eq:LA2}.

\subsection{Linkage and Capture-Recapture Model Uncertainty}

The capture-recapture model used above could be, for example, an individual decomposable graphical model, as presented in Section \ref{ss:graphicalmodels}, but it could also be the average of them as in \cite{MadiganYork97}, in which case $p_{\textsc{c}}(N\mid \bn(\bZ))$ would be given by \eqref{eq:PN_n}.  In fact, in Section \ref{s:ElSalvadorGeneralApplication} we present an application of our linkage-averaging strategy using the model of \cite{MadiganYork97} and the  Bayesian partitioning approach of \cite{Sadinle14}.  In fact, under the methodology of \cite{MadiganYork97} we can actually write the linkage-averaged posterior of $N$ as
\begin{equation}\label{eq:pNm}
p_{\textsc{la}}(N)=\sum_{\bZ}\sum_m p_{\textsc{c}}(N\mid \bn(\bZ),m) p_{\textsc{c}}(m\mid \bn(\bZ)) p_{\textsc{l}}(\bZ\mid \bX),
\end{equation}
where $m$ ranges over the non-saturated decomposable graphical models for the contingency table, and 
\begin{equation*}
p_{\textsc{c}}(m\mid \bn(\bZ)) = \frac{p(m)\sum_{N}\mathcal{L}_m(N\mid \bn(\bZ))p(N)}{\sum_{m} \sum_N \mathcal{L}_m(N\mid \bn(\bZ))p(N)p(m)},
\end{equation*}
with 
$\mathcal{L}_m(N\mid \bn(\bZ))=P(\bn(\bZ)\mid N,m)$ as given by \eqref{eq:Pn_Nm}.  Expression \eqref{eq:pNm} explicitly shows the contribution of coreference uncertainty and capture-recapture model uncertainty on the overall population size posterior.  
In fact,  we can decompose the overall posterior variance as follows
\begin{align}\label{eq:var_decomp2}
Var(N\mid \bX)   = \ & 
    Var_{\bZ\mid \bX}[E(N\mid \bZ)] 
    +\ E_{\bZ\mid \bX}\{Var_{m\mid \bZ}[E(N\mid \bZ,m)]\} \\
		& +\ E_{\bZ\mid \bX}\{E_{m\mid \bZ}[Var(N\mid \bZ,m)]\} \nonumber,
\end{align}
where the first term of the sum can be directly attributed to linkage uncertainty, the second term to model uncertainty in the population size estimation stage, and the remaining variability is intrinsic from Bayesian approaches to estimate $N$. 

We note that in principle we could also average the models of \cite{MadiganYork97} with that of \cite{Manrique16}, or with any other that satisfies the conditions discussed in Section \ref{s:pop_size}, but we do not pursue that here.

\subsection{Implications for Model Exploration and Data Confidentiality Protection}

The strategy presented here allows the linkage and the population size estimation to be carried out in two separate stages, while still leading to proper Bayesian inferences.  This has important practical implications, as the linkage can be performed as if it was the final goal, the population size estimation is standard given each coreference partition, and the combination of the two stages  through linkage-averaging is simple.

Regarding model exploration, in principle an analyst would have to obtain a new posterior $p(N,\bZ\mid \bX)$ for each different capture-recapture model being considered.  For example, the approaches of \cite{TancrediLiseo11} and \cite{LiseoTancredi11} rely on a specific capture-recapture model in the case of two lists.  Under our approach, however, we can reuse the results from the linkage step to obtain  different linkage-averaged estimates for each different capture-recapture model.  Theorem \ref{validity} implies that each such linkage-averaged posterior corresponds to a proper posterior distribution.  Not having to re-do the linkage for each different capture-recapture model is certainly an important practical advantage.  

Our two-stage strategy also indicates that the linkage and the population size estimation can be done by different analysts.  This is relevant in contexts where one needs to protect the confidentiality of the lists and the privacy of the individuals, given that the linkage can be carried out by a small trusted team, and then the linkage results, in the form of draws from $p_{\textsc{l}}(\bZ\mid \bX)$, can be transferred to subsequent analysts without having to reveal personally identifiable information used for the linkage.  


\section{Estimating Mortality Levels in the Salvadoran Civil War}\label{s:ElSalvadorGeneralApplication}

A common goal in quantitative human rights research is to estimate the total number of civilian casualties that occurred during a war.  For this purpose, multiple-systems estimation is frequently used after different lists of casualties are combined via record linkage techniques, but typically the linkage uncertainty is ignored.  \cite{LumPriceBanks13} provide a comprehensive review of such applications.  In this section we study a case from the civil war that the Central American republic of El Salvador endured between 1980 and 1991.  Our goal is to combine three data sources on civilian killings that were collected by three different organizations, and then use those results to obtain different multiple-systems estimates of the total number of civilian killings.  We focus on the region ({\em departamento}) of the capital city, San Salvador.

\subsection{Description of the Datafiles}

The first two datafiles that we consider contain reports on civilian killings collected {\em during} the civil war.  The first data source was put in electronic form by the Los Angeles-based nongovernmental organization \emph{El Rescate}, from reports that had been published periodically during the civil war by the project \emph{Tutela Legal} of the Archdiocese of San Salvador \citep{Howland08}.  We refer to this source as \textit{El Rescate / Tutela Legal} (ER-TL, 1364 records from San Salvador). The second data source comes from the \textit{Salvadoran Human Rights Commission} ({\em Comisi\'on de Derechos Humanos de El Salvador --- CDHES}, 285 records from San Salvador), which directly collected testimonials on human rights violations between 1979 and 1991 \citep{Ball00ElSalvador}.  For both datafiles, the characteristics of their collection make us believe that they should contain only small amounts of duplication, if any \citep{Sadinle17}.  

The third datafile was collected by the \textit{United Nations Truth Commission for El Salvador} (UNTC, 440 records from San Salvador),
 between 1992 and 1993, {\em after} the civil war ended \citep{CTElSalvador93}.  Given that most of the reports to the UNTC refer to killings that occurred several years before 1992, it is reasonable to expect the information in this datafile to be less reliable compared with ER-TL and CDHES, since individuals reporting casualties might not have had accurate recollections of the time and place of the events.  Non-trivial duplicates arise in this datafile from reports of multiple family members and acquaintances of a single victim.

\subsection{Record Linkage and Duplicate Detection}\label{ss:ElSalvador_linkage}

\subsubsection{Datafile Standardization, Filtering Non-Coreferent Pairs, and Comparison Data}

The three datafiles used in this section have the following fields in common: given and family names, date and place/municipality of death.  
Our standardization of names and construction of comparison data are as described in the application of \cite{Sadinle14}.  Table \ref{t:compdata} summarizes the construction of levels of disagreement.  Since the datafiles are small enough,  we computed comparison data for all $\binom{2,089}{2} = $ 2,180,916 record pairs (the set $\mathcal{P}$).  We then formed the set $\mathcal{C}$ of candidate coreferent pairs by fixing as non-coreferent all the pairs that have disagreement level three in either given name or family name.  This leads to only $|\mathcal{C}|=699$ candidate pairs, which involve only 775 records. 

\begin{table}[t]
\centering
\footnotesize{
  \begin{minipage}[b]{1\textwidth}
 \caption{Construction of levels of disagreement for lists from El Salvador.}\label{t:compdata}
\centering
\begin{tabular}{lccccc}
  \hline\\[-8pt]
       &                     & \multicolumn{4}{c}{Levels of Disagreement}\\
          \cline{3-6}\\ [-6pt]
  Field & Similarity Measure & $0$ & $1$ & $2$ & $3$ \\
  \hline\\[-8pt]
  Given Name & Normalized Levenshtein\footnote{\label{note1} Modification of \cite{Sadinle14} to account for Hispanic naming conventions.} &  0 & $(0,0.25]$ & $(0.25,0.5]$ & $(0.5,1]$ \\
  Family Name & Normalized Levenshtein\footnoteref{note1} & 0 & $(0,0.25]$ & $(0.25,0.5]$ & $(0.5,1]$ \\
  Year of Death & Absolute Difference &  0 & 1 & 2--3 & 4+ \\
  Month of Death& Absolute Difference & 0 & 1 & 2--3 & 4+ \\
  Day of Death & Absolute Difference &  0 & 1--2 & 3--7 & 8+ \\
  Place of Death & Binary Comparison & Agree & Disagree &&\\
  \hline
\end{tabular}
\end{minipage}
}
\end{table}

\subsubsection{Prior Specification}

We followed the general guidelines presented in Section \ref{ss:comparison-based} and used uniform priors on $[0,1]$ for all the $u_{fl}$ parameters.  For the $m_{fl}$ parameters, we used flat priors in the intervals $[\lambda_{fl},1]$ for the truncation points given in Table \ref{t:prior_3files}.  These priors indicate our belief that coreferent pairs are very likely to have exact agreements, although we still expect a considerable amount of error in the fields.    Finally, the prior for the field day of death has low truncation points in general,  since we believe this field to be unreliable.  

\begin{table}[h]
\centering
\footnotesize{
  \begin{minipage}[b]{1\textwidth}
 \caption{Prior truncation points $\lambda_{fl}$ for the $m_{fl}$ parameters in the joint duplicate detection and record linkage for three datafiles from El Salvador.}\label{t:prior_3files}
\centering
\begin{tabular}{clccccccc}
  \hline\\[-8pt]
&       &    \multicolumn{2}{c}{Name} && \multicolumn{3}{c}{Date of Death} & \\
          \cline{3-4}\cline{6-8}\\ [-8pt]
& $l$  & Given	& Family && Year & Month & Day & Municipality \\
  \hline\\[-8pt]
& 0 &0.95 &	 0.95 &&	 0.90 &	 0.80 &	 0.70 &	 0.80\\ 
& 1 &0.99 &	 0.99 &&	 0.95 &	 0.90 &	 0.70 &	   --  \\
& 2 &0.99 &	 0.99 &&	 0.99 &	 0.99 &	 0.70	&    --  \\
\\[-8pt]
  \hline
\end{tabular}
\end{minipage}
}
\end{table}

\subsubsection{Gibbs Sampler Implementation}\label{ss:Gibbs_DDRL}

We ran 10,000 iterations of the Gibbs sampler of \cite{Sadinle14}.  The runtime using an implementation in R with parts written in C language was of 35 seconds on a laptop with a processor Intel Core i7-4900MQ.  Convergence of the chain was checked using functions of the partitions.  We found the  number of killings reported 1, 2, and 3 times according to each partition in the chain.  
The traceplots of these chains (not shown here) indicate that they seem to have converged rather quickly, and their autocorrelation functions indicate that  there are not large autocorrelations in the chain.  Similar results were obtained when we explored the number of different killings in the datafiles according to the partitions in the chain.  Based on these diagnostics we discarded the first 1,000 iterations and kept one draw each five iterations.  After this thinning, the autocorrelation plots (not shown here) did not suggest the existence of remaining autocorrelations of any order.  For each of the previously explored chains we also computed Geweke's convergence diagnostic as implemented in the R package \verb"coda" \citep{CODA}.  The Geweke's Z-scores indicated that it is reasonable to treat these chains as drawn from their stationary distributions.  We also explored the marginal probabilities that pairs of records are coreferent for the pairs in the set $\mathcal{C}$ of candidate pairs.  For each pair in $\mathcal{C}$, and for each partition in the chain, we checked whether the pair appeared together in the partition.  For each of these binary chains we computed Geweke's convergence diagnostic, and we found that all the Z-scores range around the usual values of a standard normal random variable, which indicates that it is reasonable to assume that these chains were obtained from their stationary distributions.

\subsection{Linkage-Averaged Posterior Estimates of the Total Number of Killings}\label{s:PopSize_ElSalvador}

The draws from the posterior of the coreference partition can be directly used to obtain inferences on different quantities of interest.  For example, computing the size of each partition gives us posterior draws of the number of different reported killings, which in this case lead to a 99\% credible interval of [1892, 1906], and a posterior mean of 1900.  This can be seen as an estimated lower bound on the total number of killings.  In Table \ref{t:CTuncertain} we also present the marginal posterior distribution of number of killings following each of the different inclusion patterns, $n_{111}, \dots, n_{100}$.  The remainder of the section is devoted to using the posterior draws of the coreference partition to derive estimates of the total number of killings using different capture-recapture models. 

\begin{table}[h]
 \centering
  \begin{minipage}[b]{1\textwidth}
 \def\~{\hphantom{0}}
  \caption{Marginal posterior distributions of the frequencies of inclusion patterns.} \label{t:CTuncertain}
  \begin{tabular*}{\columnwidth}{c@{\extracolsep{\fill}}c@{\extracolsep{\fill}}c@{\extracolsep{\fill}}c@{\extracolsep{\fill}}c@{\extracolsep{\fill}}c@{\extracolsep{\fill}}c@{\extracolsep{\fill}}}
\hline\\ [-8pt]
   &  & \multicolumn{2}{c}{In UNTC}  & & \multicolumn{2}{c}{Out UNTC} \\
 \cline{3-7}\\ [-8pt]
 ER-TL & & In CDHES & Out CDHES &  & In CDHES & Out CDHES\\
	\cline{1-7} \\ [-8pt]
  In & & \begin{tabular}{c}\includegraphics[width=0.15\columnwidth]{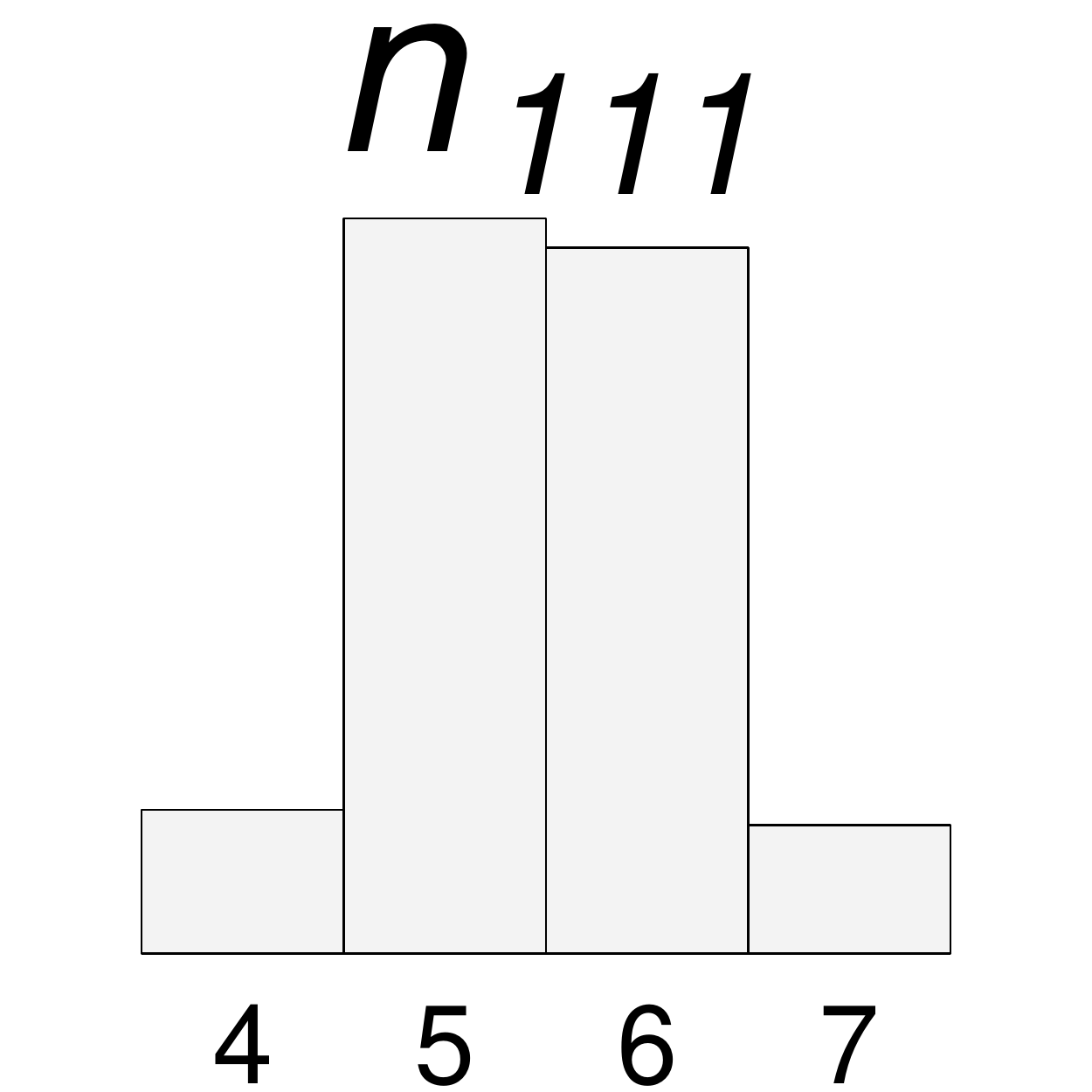}\end{tabular} & \begin{tabular}{c}\includegraphics[width=0.15\columnwidth]{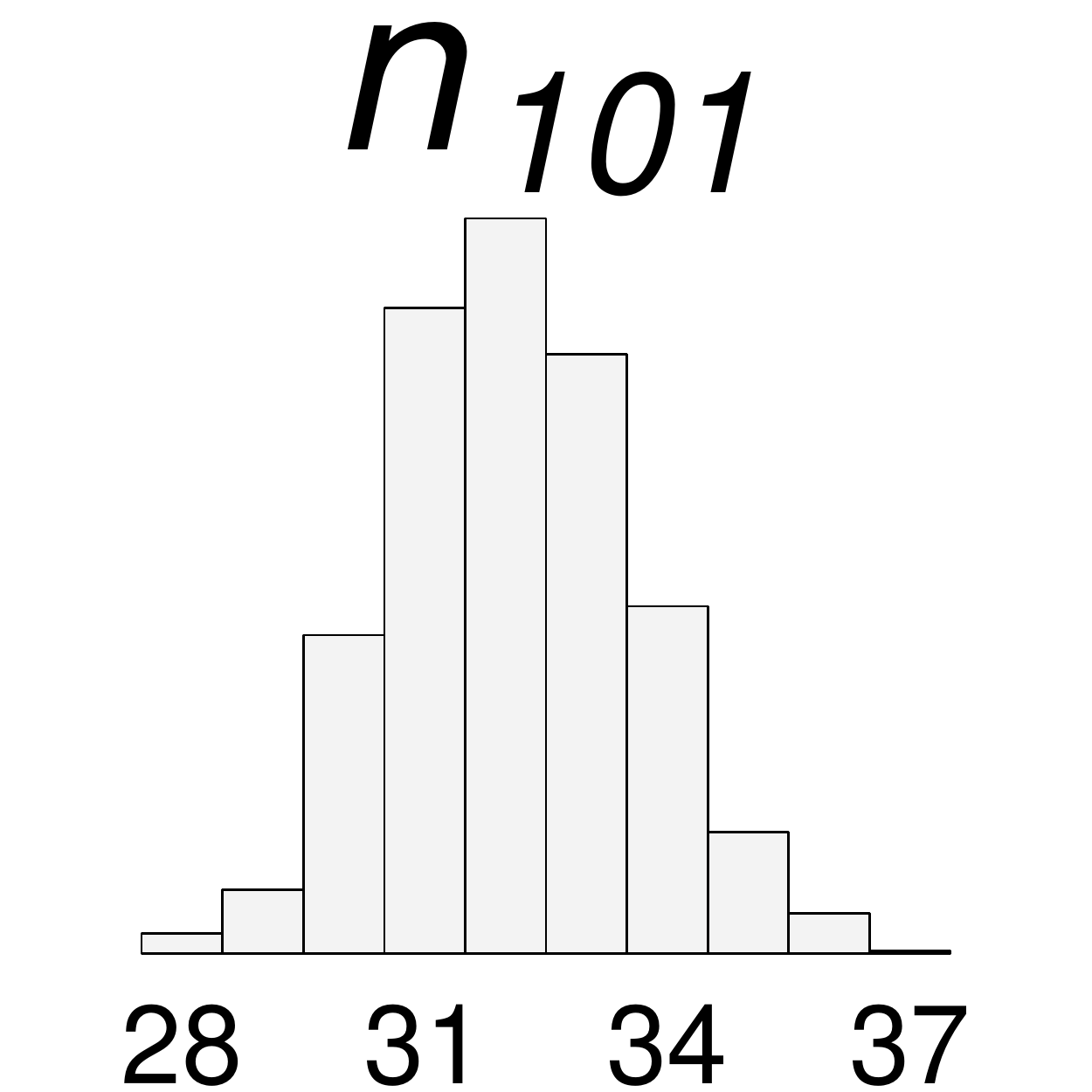}\end{tabular} & & \begin{tabular}{c}\includegraphics[width=0.15\columnwidth]{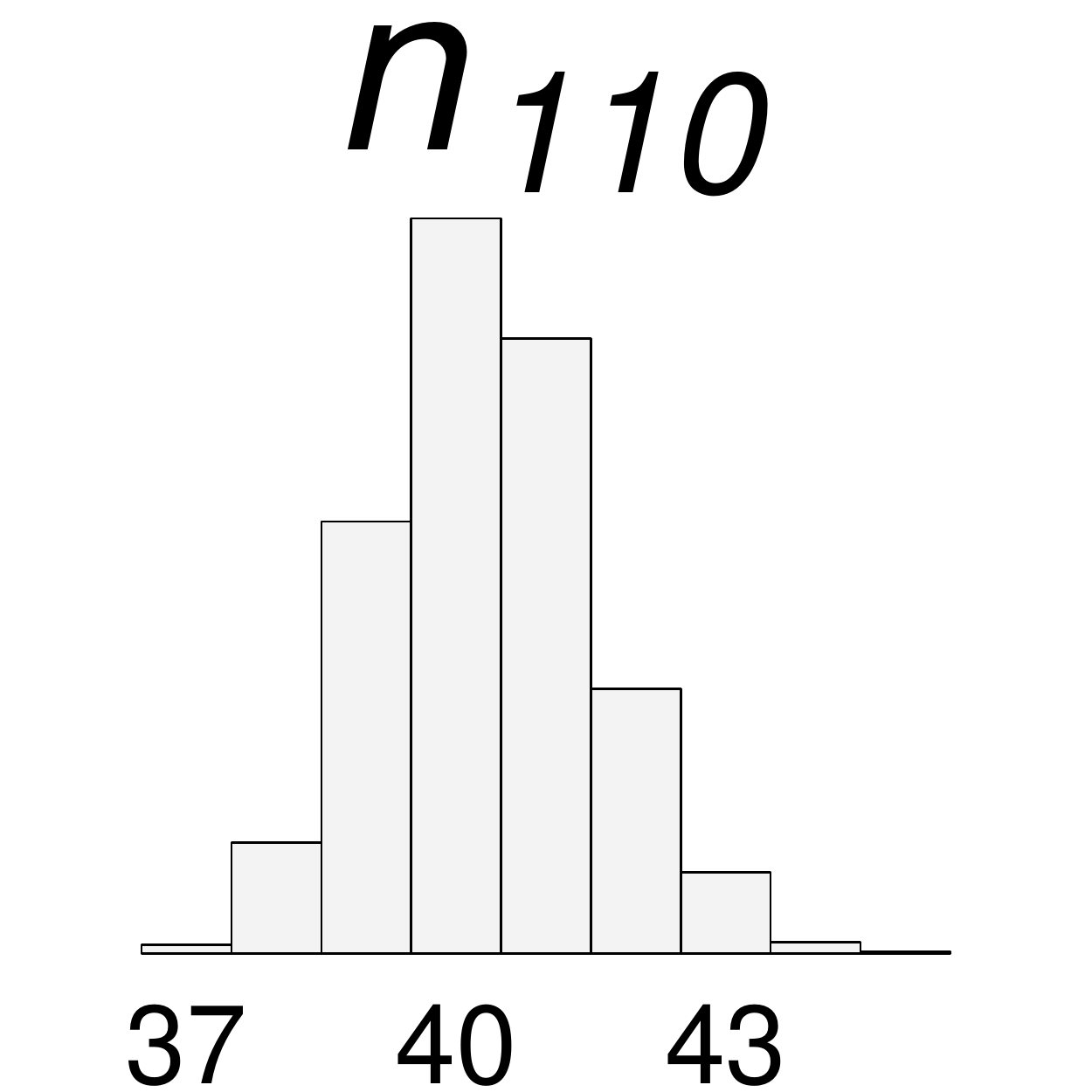}\end{tabular} & \begin{tabular}{c}\includegraphics[width=0.15\columnwidth]{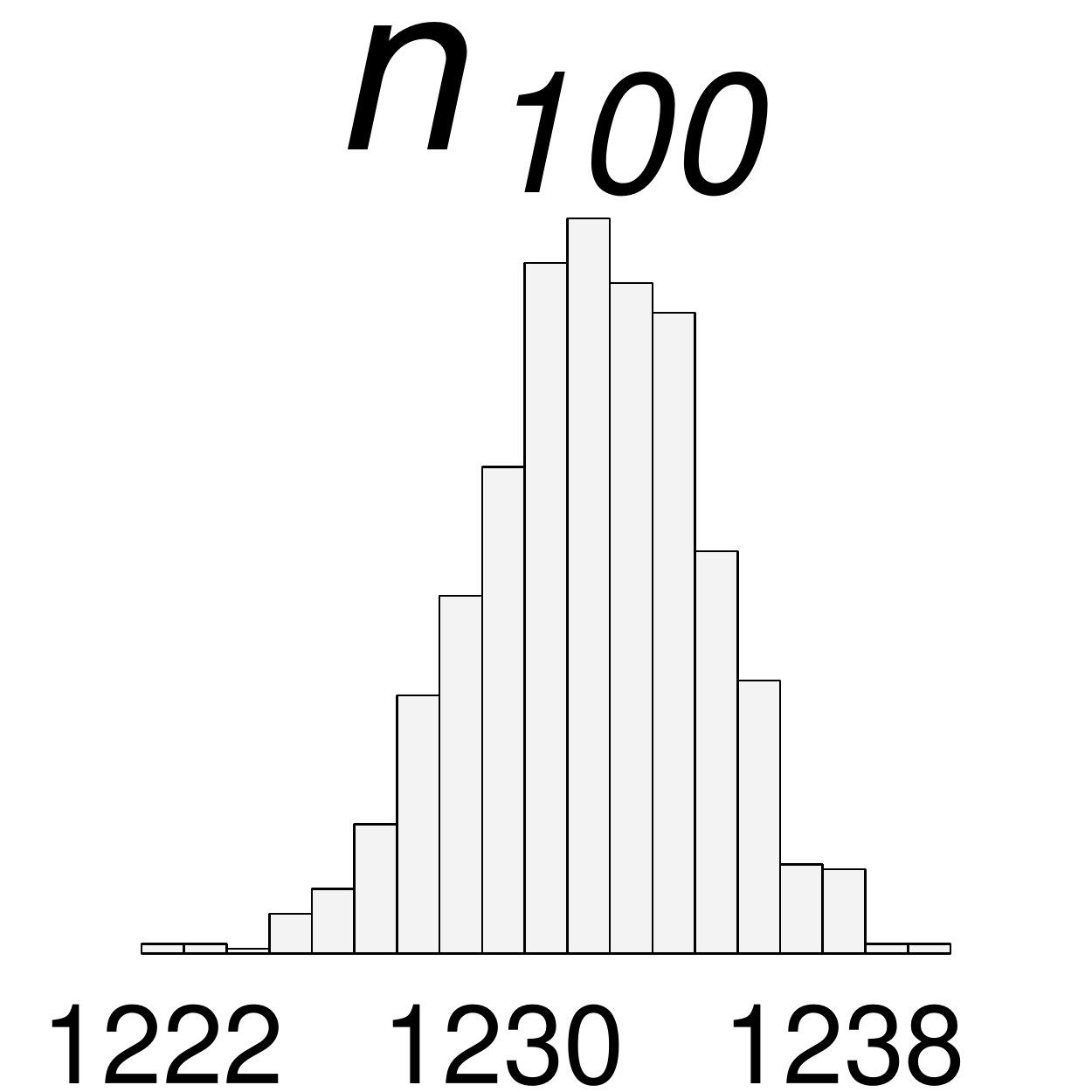}\end{tabular} \\
  Out & & \begin{tabular}{c}\includegraphics[width=0.15\columnwidth]{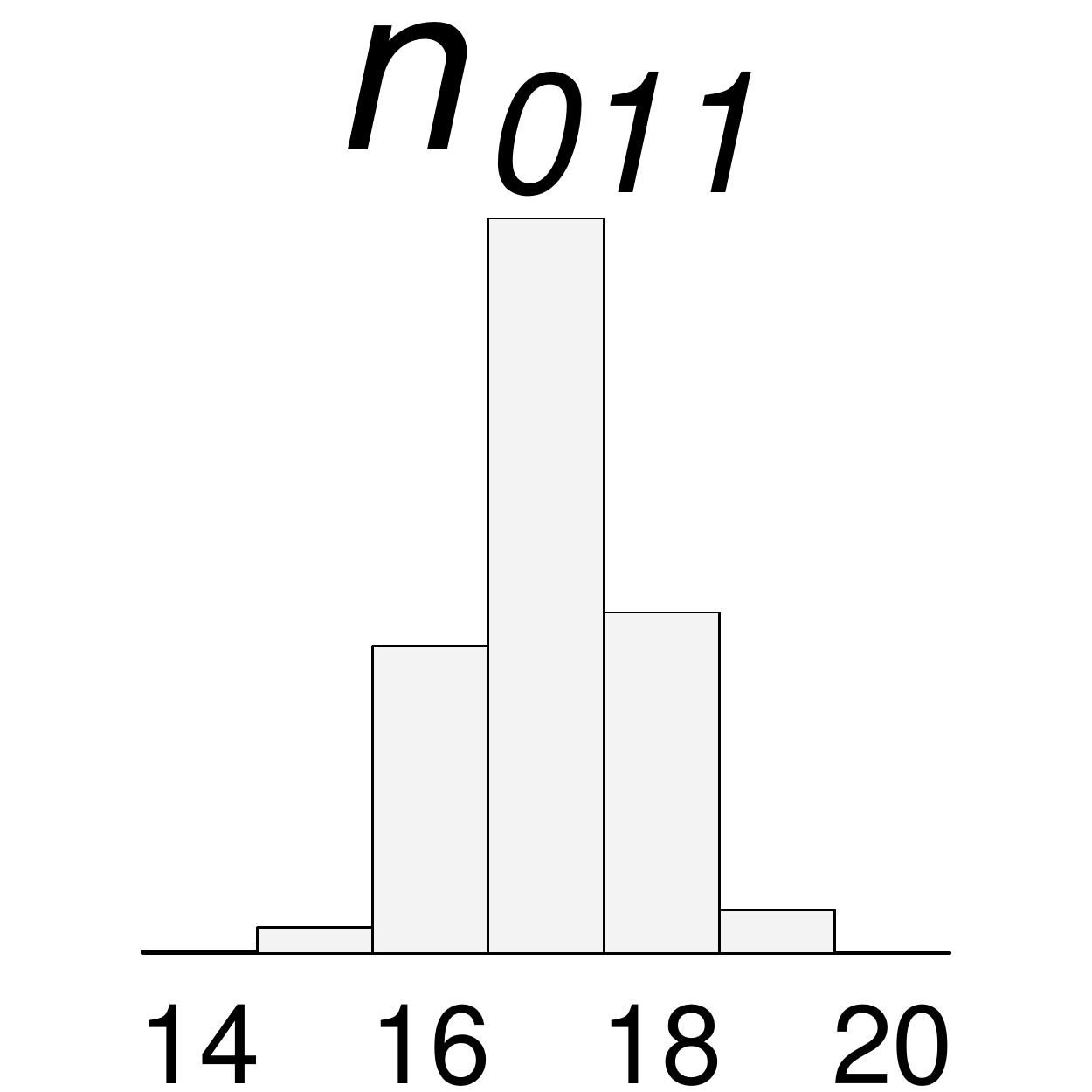}\end{tabular} & \begin{tabular}{c}\includegraphics[width=0.15\columnwidth]{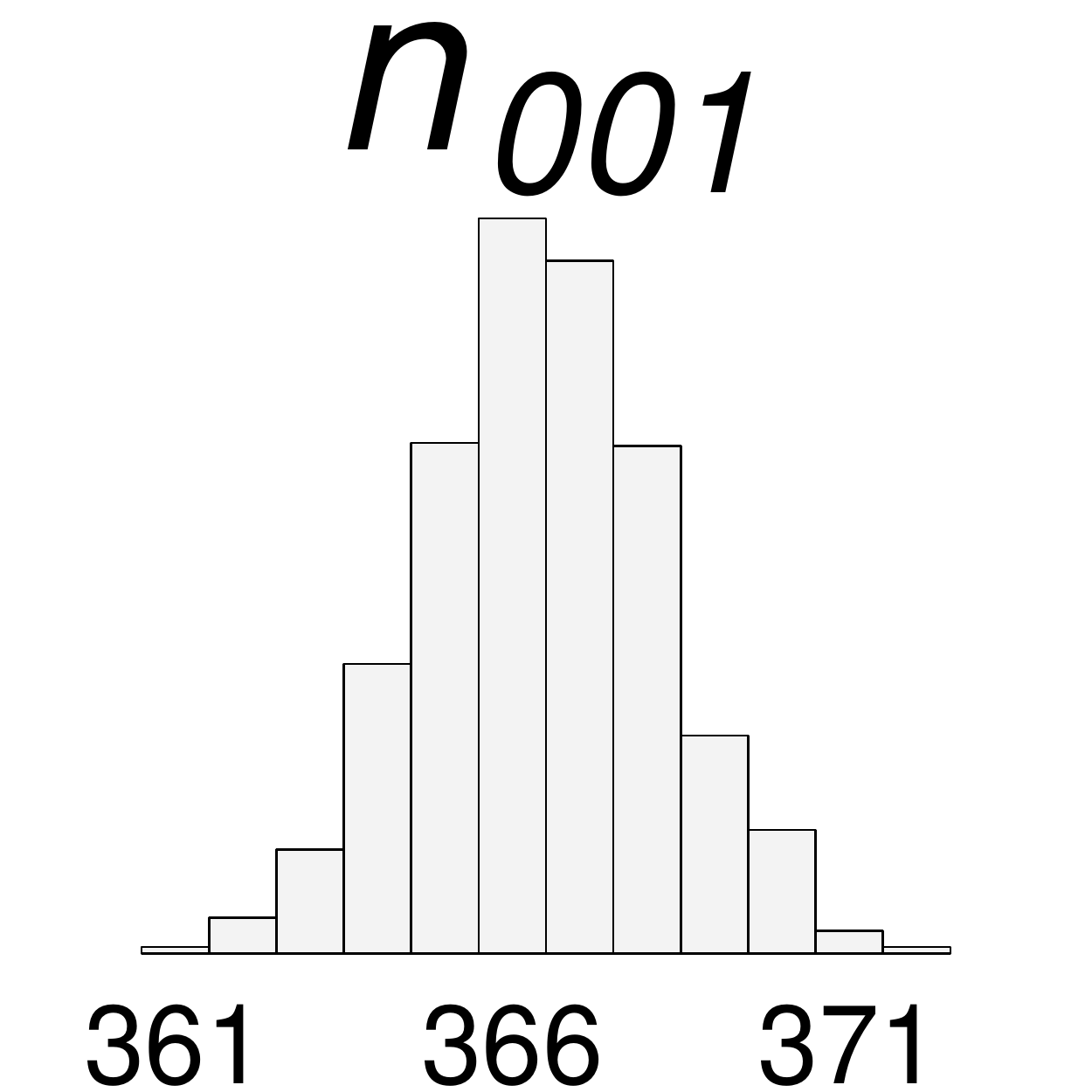}\end{tabular} & & \begin{tabular}{c}\includegraphics[width=0.15\columnwidth]{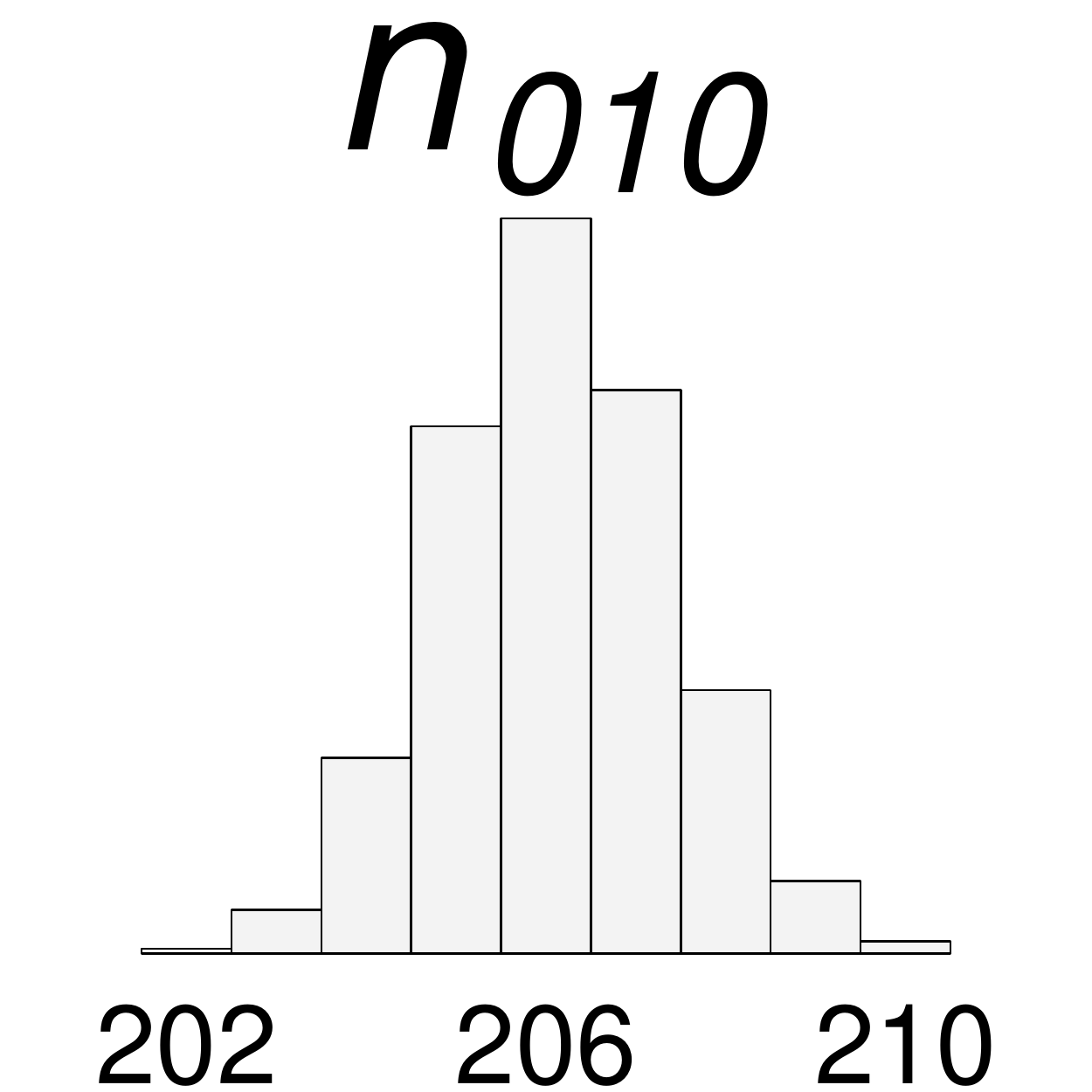}\end{tabular} & $-$ \\[-1pt]\hline
\end{tabular*}\vskip18pt
\end{minipage}
\end{table}

\subsubsection{Two-Sample Estimates}

In Section \ref{s:lapse} we mentioned that the subsequent capture-recapture model does not necessarily have to use all the lists combined in the linkage step.  For example, the linkage step may have included datafiles whose collection make the assumptions in the capture-recapture model implausible.  An example in the context of our application would be a list of the victims that belonged to a given organization; in that case, non-members of the organization would have zero probability of being included in the list, by definition. 

In this section, we use the results of the linkage step to derive estimates of the population size based only on the inclusion patterns for pairs of lists.  For example, using only the first two data sources to estimate the population size, we need to compute $n_{11+}(\bZ)=n_{110}(\bZ)+n_{111}(\bZ)$, $n_{10+}(\bZ)=n_{100}(\bZ)+n_{101}(\bZ)$, and $n_{01+}(\bZ)=n_{010}(\bZ)+n_{011}(\bZ)$, for each coreference partition labeling $\bZ$ in the posterior sample from the linkage step, and use these as sufficient statistics for the capture-recapture model.  This modeling approach does not take advantage of the additional piece of information $n_{001}(\bZ)$.  With only two sources, we are limited to the capture-recapture model that assumes independence of the inclusion of the victims in the data sources, as the counts $n_{11+}$, $n_{10+}$, and $n_{01+}$ do not contain enough information to estimate this  dependence.  A possible alternative would be to pre-specify a degree of dependence between the sources, for example as discussed in \cite{EricksenKadaneTukey89}, but we do not pursue that avenue here.  In the model of independence, the modeling approach of \cite{MadiganYork97} corresponds to the approach of \cite{Castledine81}.  

In Table \ref{t:2sample} we present summaries of each linkage-averaged posterior of $N$ obtained using the different possible pairs of datafiles for the estimation of the population size $N$. The fifth column in that table shows the percentage contribution of the linkage variability towards the overall posterior variability of the population size, derived from \eqref{eq:var_decomp}.  For some of the models this contribution can be quite small, meaning that in such cases obtaining a posterior estimate of the inclusion patterns' frequencies and fixing those to estimate $N$ would lead to similar inferences compared with those from linkage-averaging.  However, we only obtain this information {\em after} we compute the variance decomposition in \eqref{eq:var_decomp}.

\begin{table}[h]
\centering
  \begin{minipage}[b]{1\textwidth}
  \caption{Linkage-averaging for two-sample estimates of $N$. $\hat N$: expected value computed from $p_{\textsc{la}}(N)$. CI: credible interval.  The plots in the second column have the same horizontal and vertical scales.}
\label{t:2sample}\end{minipage} 
\begin{tabular}{ccccc}
\hline\\ [-8pt]
 & \multicolumn{4}{c}{Linkage-Averaging}\\
          \cline{2-5}\\ [-8pt]
Lists & $p_{\textsc{la}}(N)$ & $\hat N$ & 99\% CI & Linkage Var.\\
\hline
ER-TL, CDHES &\begin{tabular}{c}\includegraphics[trim = 0cm 0cm 0cm 0cm, width=0.25\columnwidth]{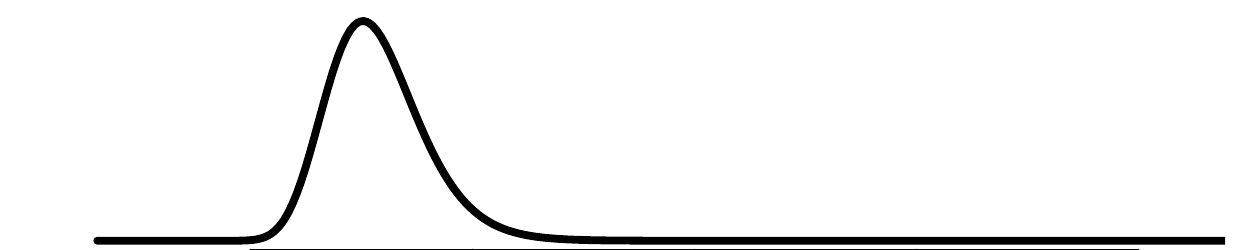}\end{tabular}& 7,852 &[5613, 11367]& 5.06\% \\
ER-TL, UNTC  &\begin{tabular}{c}\includegraphics[trim = 0cm 0cm 0cm 0cm, width=0.25\columnwidth]{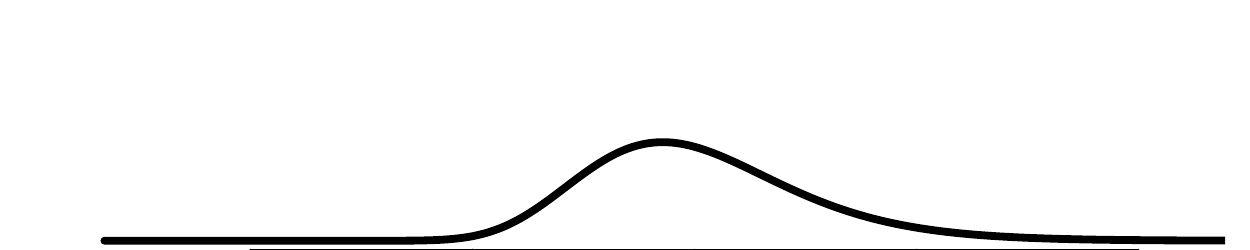}\end{tabular}& 15,082 &[10156, 23218]& 6.68\% \\
CDHES, UNTC  &\begin{tabular}{c}\includegraphics[trim = 0cm 0cm 0cm 0cm, width=0.25\columnwidth]{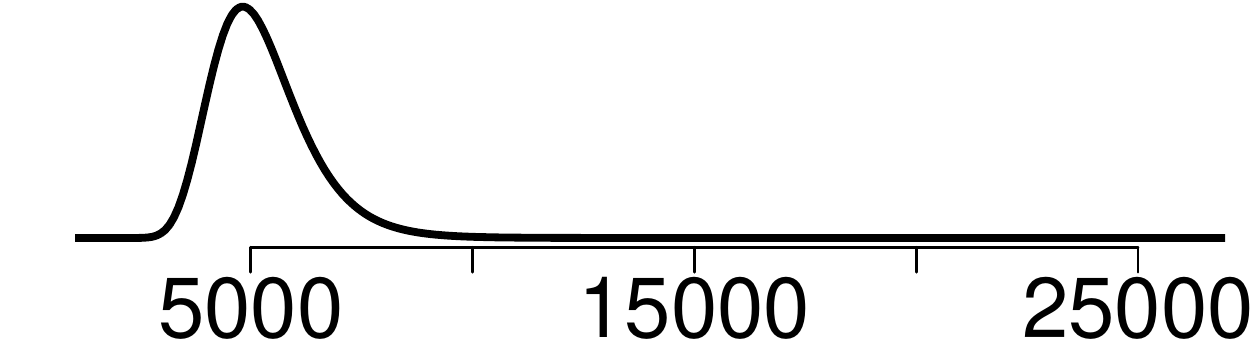}\end{tabular}& 5253 &[3212, 9113]& 2.74\%\\
\hline
\end{tabular}
\end{table}

\subsubsection{Three-Sample Estimates from Individual Graphical Models}

We now obtain estimates of $N$ based on each of the individual graphical models presented in Section \ref{ss:graphicalmodels}.  Fixing one such model to estimate $N$ could arise in a context where one can conjecture the dependence graph based on domain knowledge, such as knowledge of collaboration, affinity, or antagonism between institutions collecting the data. 

We summarize the linkage-averaged posteriors obtained using each individual graphical model in Table \ref{t:3var_graphs}.  Similarly as for the two-sample estimates, we can see that the relative contribution of the linkage uncertainty towards the posterior uncertainty around $N$ can be quite small, meaning that the importance of accounting for linkage uncertainty ends up depending on the specific model.  Unfortunately, there does not seem to be a way to tell in advance if the linkage uncertainty is going to have a big impact on the estimation of $N$.  

\begin{table}[h]
\centering
  \begin{minipage}[b]{1\textwidth}
  \caption{Summaries of linkage-averaging for three-sample population size estimates using individual graphical models. $\hat N$: expected value computed from $p_{\textsc{la}}(N)$. CI: credible interval.  The plots in the third column have the same horizontal and vertical scales.  The data sources are 1: ER-TL, 2: CDHES, 3: UNTC.}
\label{t:3var_graphs}\end{minipage} 
\begin{tabular}{ccccccc}
\hline\\ [-8pt]
\multicolumn{2}{c}{Graphical Model} && \multicolumn{4}{c}{Linkage-Averaging}\\
\cline{1-2}          \cline{4-7}\\ [-8pt]
Notation & Graph && $p_{\textsc{la}}(N)$ & $\hat N$ & 99\% CI & Linkage Var.\\
\hline\\ [-8pt]
$[1][2][3]$&  
\begin{tabular}{c}
\begin{tikzpicture}[thick,every node/.style={draw,circle,scale=.4}]
\node (1) {\textbf{1}};
\node (2) [right = 0.15cm of 1] {\textbf{2}};
\node (3) [below = 0.1cm of 2,xshift=-.5cm] {\textbf{3}};
\end{tikzpicture}\end{tabular} 
&\multicolumn{2}{c}{\begin{tabular}{c}\includegraphics[width=0.2\columnwidth]{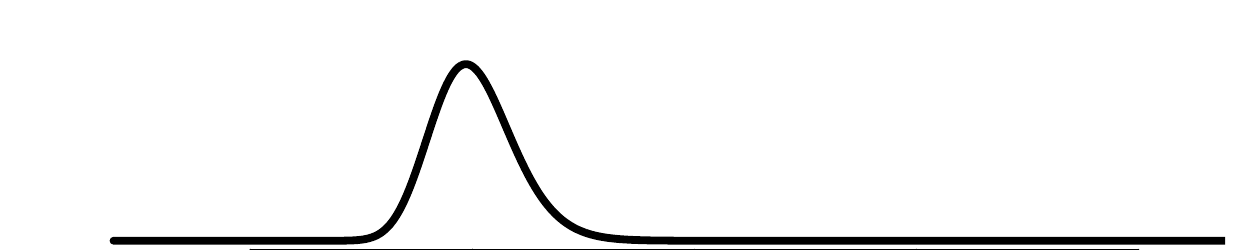}\end{tabular}}& 10041& [7967, 12847]& 4.84\% \\
$[1,2][3]$&
\begin{tabular}{c}
\begin{tikzpicture}[thick,every node/.style={draw,circle,scale=.4}]
\node (1) {\textbf{1}};
\node (2) [right = 0.15cm of 1] {\textbf{2}};
\node (3) [below = 0.1cm of 2,xshift=-.5cm] {\textbf{3}};
\draw (1) -- (2);
\end{tikzpicture}\end{tabular}	
&\multicolumn{2}{c}{\begin{tabular}{c}\includegraphics[width=0.2\columnwidth]{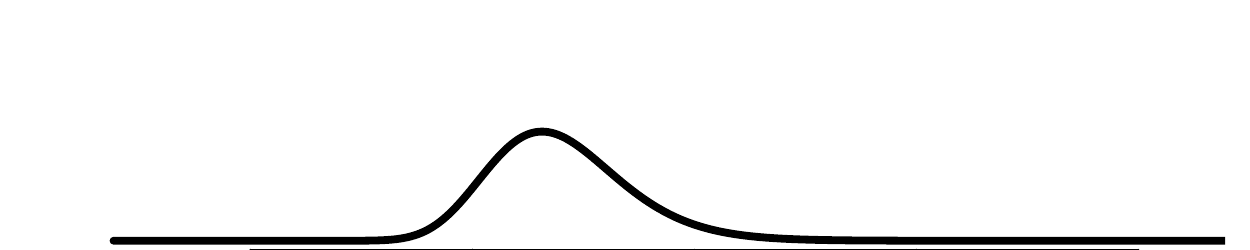}\end{tabular}}& 11986& [8743, 16881]& 5.02\% \\
$[1,3][2]$&
\begin{tabular}{c}
\begin{tikzpicture}[thick,every node/.style={draw,circle,scale=.4}]
\node (1) {\textbf{1}};
\node (2) [right = 0.15cm of 1] {\textbf{2}};
\node (3) [below = 0.1cm of 2,xshift=-.5cm] {\textbf{3}};
\draw (1) -- (3);
\end{tikzpicture}\end{tabular}
&\multicolumn{2}{c}{\begin{tabular}{c}\includegraphics[width=0.2\columnwidth]{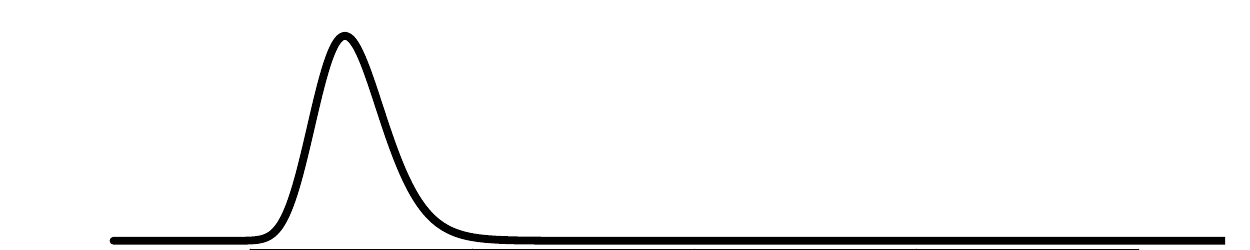}\end{tabular}}& 7331& [5588, 9881]& 3.53\%\\
$[1][2,3]$&
\begin{tabular}{c}
\begin{tikzpicture}[thick,every node/.style={draw,circle,scale=.4}]
\node (1) {\textbf{1}};
\node (2) [right = 0.15cm of 1] {\textbf{2}};
\node (3) [below = 0.1cm of 2,xshift=-.5cm] {\textbf{3}};
\draw (2) -- (3);
\end{tikzpicture}\end{tabular}
&\multicolumn{2}{c}{\begin{tabular}{c}\includegraphics[width=0.2\columnwidth]{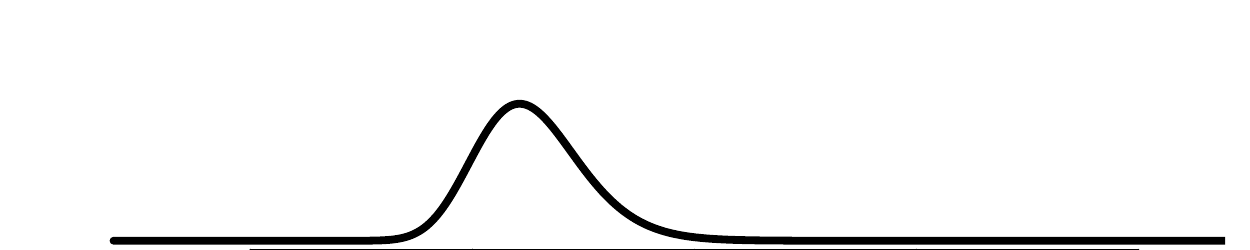}\end{tabular}}& 11339& [8702, 15061]& 5.45\%\\
$[1,2][1,3]$&
\begin{tabular}{c}
\begin{tikzpicture}[thick,every node/.style={draw,circle,scale=.4}]
\node (1) {\textbf{1}};
\node (2) [right = 0.15cm of 1] {\textbf{2}};
\node (3) [below = 0.1cm of 2,xshift=-.5cm] {\textbf{3}};
\draw (1) -- (2);
\draw (3) -- (1);
\end{tikzpicture}\end{tabular}
&\multicolumn{2}{c}{\begin{tabular}{c}\includegraphics[width=0.2\columnwidth]{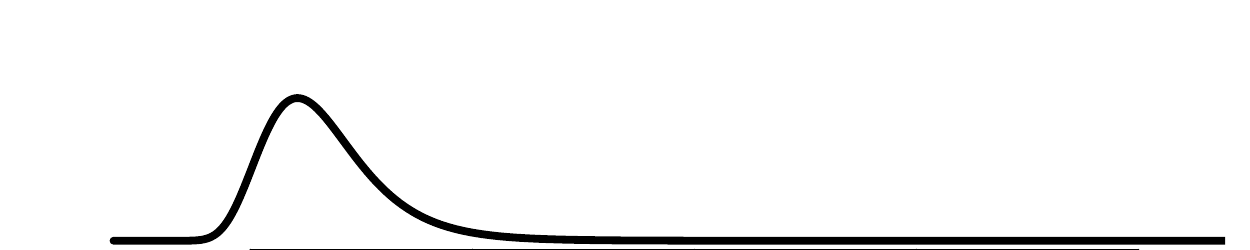}\end{tabular}}& 6637& [4319, 11512]& 3.75\%\\
$[1,2][2,3]$&
\begin{tabular}{c}
\begin{tikzpicture}[thick,every node/.style={draw,circle,scale=.4}]
\node (1) {\textbf{1}};
\node (2) [right = 0.15cm of 1] {\textbf{2}};
\node (3) [below = 0.1cm of 2,xshift=-.5cm] {\textbf{3}};
\draw (1) -- (2);
\draw (2) -- (3);
\end{tikzpicture}\end{tabular}
&\multicolumn{2}{c}{\begin{tabular}{c}\includegraphics[width=0.2\columnwidth]{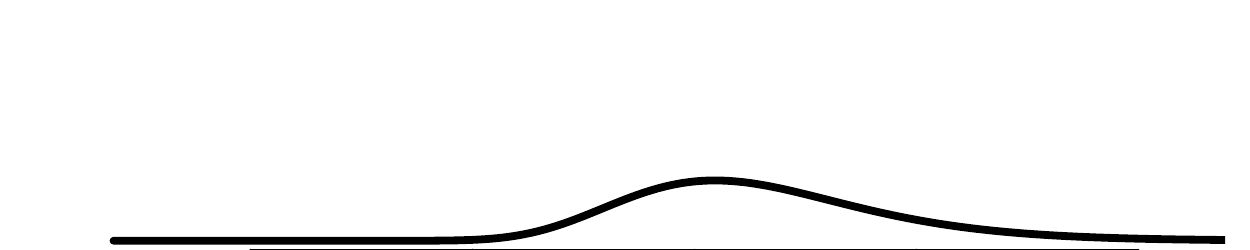}\end{tabular}}& 16466& [10800, 26043]& 7.09\%\\
$[1,3][2,3]$&
\begin{tabular}{c}
\begin{tikzpicture}[thick,every node/.style={draw,circle,scale=.4}]
\node (1) {\textbf{1}};
\node (2) [right = 0.15cm of 1] {\textbf{2}};
\node (3) [below = 0.1cm of 2,xshift=-.5cm] {\textbf{3}};
\draw (2) -- (3);
\draw (3) -- (1);
\end{tikzpicture}\end{tabular}
&\multicolumn{2}{c}{\begin{tabular}{c}\includegraphics[width=0.2\columnwidth]{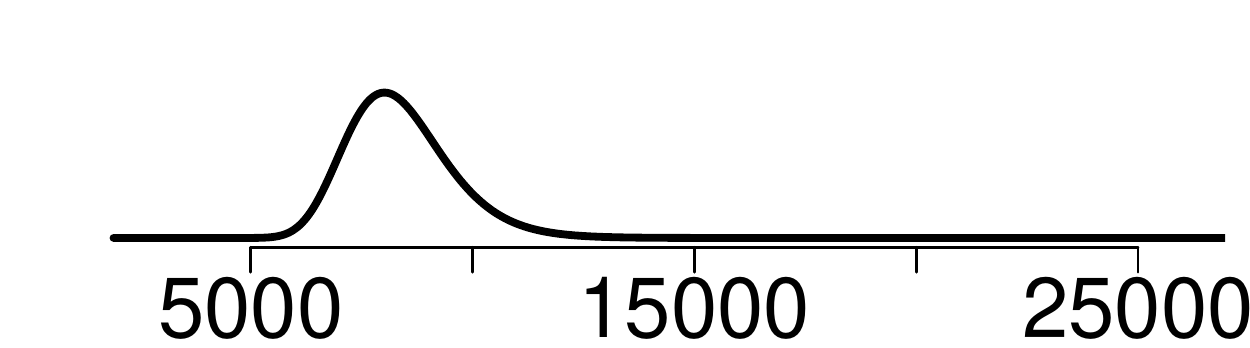}\end{tabular}}& 8380& [5989, 12233]& 4.07\%\\
\hline
\end{tabular}
\end{table}

\subsubsection{Three-Sample Estimates from \cite{MadiganYork97}}

We now use the Bayesian model averaging approach of \cite{MadiganYork97} to estimate $N$.  For each coreference partition $\bZ^{(1)},\dots,\bZ^{(d)}$, we can compute the joint posterior probability of the graphical model $m$ and the population size $N$, $p(m,N\mid \bX,\bZ^{(t)})$, which we can use to derive $p(N\mid \bX,\bZ^{(t)})$.  The gray lines in the first panel of Figure  \ref{f:postN} represent each $p(N\mid \bX,\bZ^{(t)})$ for $t=1,\dots,100$, and the black line represents the linkage-averaged posterior of $N$.  The posteriors of the number of killings derived from the individual draws $\bZ^{(1)},\dots,\bZ^{(100)}$ are somewhat similar to each other, which indicates a small contribution of the linkage uncertainty towards the overall posterior variability of $N$. According to the variance decomposition in \eqref{eq:var_decomp2}, in this case 12\% of the posterior variability is due to uncertainty in duplicate detection and record linkage.  

The second panel in Figure \ref{f:postN} shows the linkage-averaged posterior of $N$ along with $p_{\textsc{la}}(m,N\mid \bX)$, obtained from averaging $p(m,N\mid \bX,\bZ^{(t)})$ over the posterior draws $\bZ^{(1)},\dots,\bZ^{(100)}$, for the three models $m$ that have linkage-averaged posterior probabilities $p_{\textsc{la}}(m\mid \bX)>0.05$.  Denoting 1: ER-TL, 2: CDHES, and 3: UNTC, we find that the posteriors of $N$ under the models [1,3][2], [1,2][2,3], and [1,3][2,3] are concentrated around different values of $N$, which greatly increases the posterior variability of $N$.  In fact, the variance decomposition in \eqref{eq:var_decomp2} tells us that in this case 77\% of the posterior variability of $N$ is due to uncertainty on the graphical model for population size estimation.  This seems to indicate that as long as we have a good estimate of the contingency table of inclusion patterns, ignoring the linkage uncertainty in the population size estimation would not be too harmful, at least for this application.  The linkage-averaging approach leads to a posterior mean of 13,432, and a 99\% credible interval of [5627, 25404].

\begin{figure*}[t]
\centering
		\centerline{\includegraphics[trim = 0cm 0cm 3cm 0cm, width=1\linewidth]{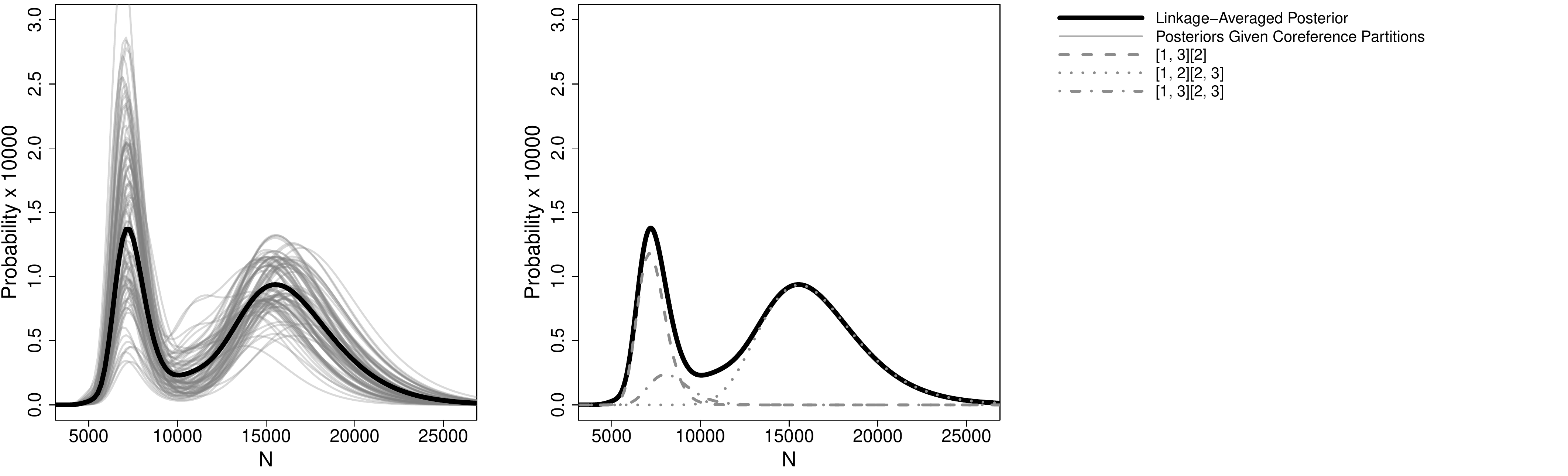}}
  \begin{minipage}[b]{1\textwidth}
  \caption{Posterior of the number of civilian killings averaging over the seven graphical models for population size estimation and over the uncertainty from record linkage and duplicate detection. The data sources are 1: ER-TL, 2: CDHES, 3: UNTC.}
\label{f:postN}\end{minipage} 
\end{figure*}

\subsubsection{Three-Sample Estimates from \cite{Manrique16}}\label{ss:Manrique}

Linkage-averaging for population size estimation can be used with any Bayesian partitioning approach to record linkage and duplicate detection, and any model for population size estimation that depends only on the capture histories' frequencies of the individuals in the lists.  We now use the linkage results described in Section \ref{ss:ElSalvador_linkage} obtained from the approach of \cite{Sadinle14}, along with the population size methodology of \cite{Manrique16}.  

For each of 100 draws $\bZ^{(1)},\dots,\bZ^{(100)}$, we obtained an MCMC sample $N^{(1,t)}, \dots, N^{(20000,t)} \sim p(N\mid \bn(\bZ^{(t)}))$, $t=1,\dots,100$, from the posterior obtained under the model of \cite{Manrique16} using the MCMC implementation of the R package \texttt{LCMCR}.  We then used the approximation \eqref{eq:LA2} of the linkage-averaged posterior of $N$.  Figure \ref{f:postN_MV} presents an approximation of each of $p(N\mid \bn(\bZ^{(t)}))$, $t=1,\dots,100$, and the approximate linkage-averaged posterior of $N$, $p_{\textsc{la}}(N\mid \bX)$.  Under this approach we obtain a posterior 99\% credible interval of [4922, 31429] and a posterior mean of 13,924.  The contribution of the linkage uncertainty to the overall posterior variability is estimated at only 6.3\%.  

\begin{figure*}[t]
\centering
		\centerline{\includegraphics[trim = 0cm 0cm 3cm 0cm, width=.7\linewidth]{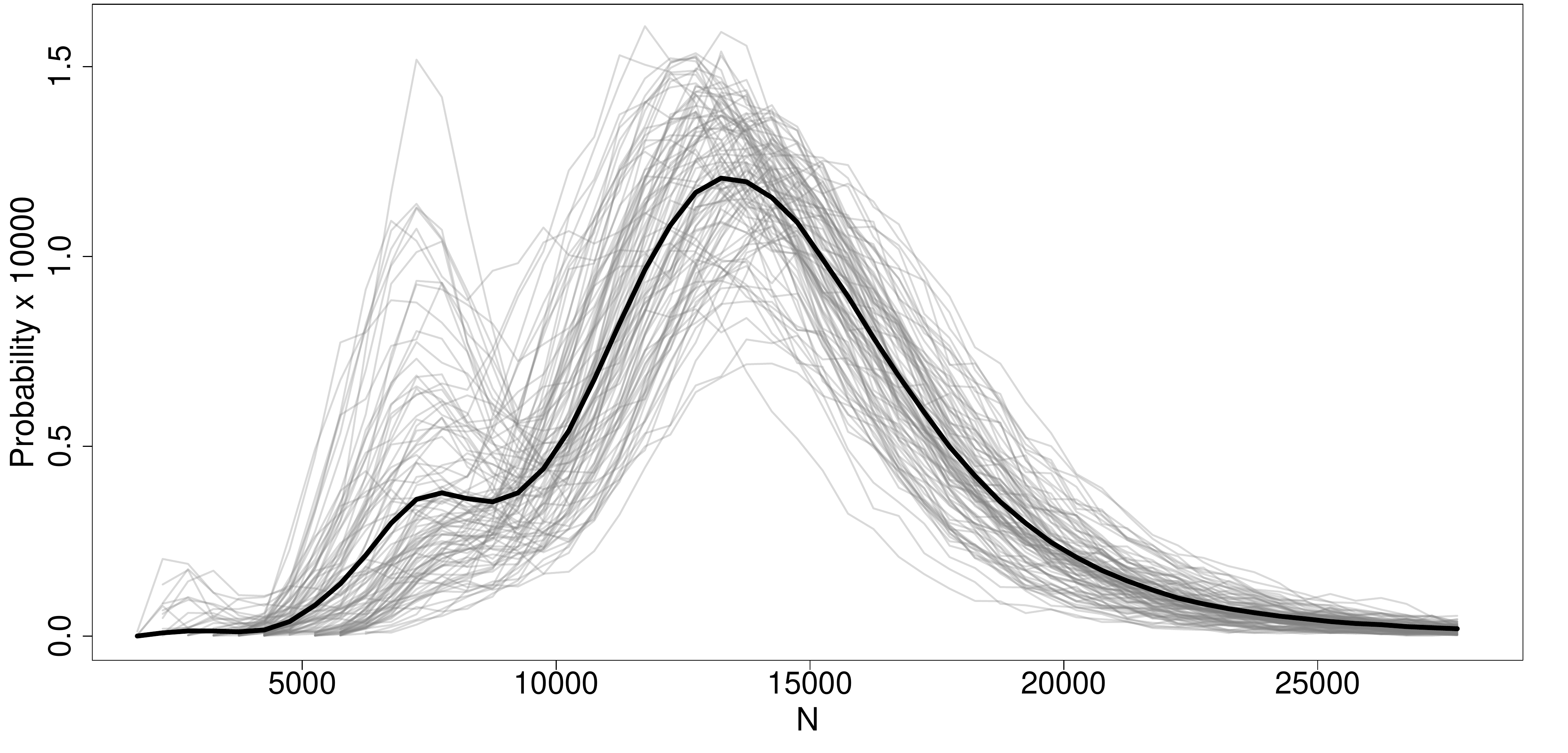}}
  \begin{minipage}[b]{1\textwidth}
  \caption{Black line: linkage-averaged posterior of the number of civilian killings in the region of San Salvador using the population size estimation methodology of \cite{Manrique16}.  Gray lines: individual posteriors of the number of killings given each coreference partition labeling $\bZ^{(t)}$, $t=1,\dots,100$.}
\label{f:postN_MV}\end{minipage} 
\end{figure*}

\subsubsection{Estimates Using Mixture-Model Approach to Record Linkage}

We finally present the results obtained using a more traditional mixture model approach to record linkage  \citep[e.g.][]{FellegiSunter69,Winkler88, Jaro89,LarsenRubin01,Elmagarmidetal07,HerzogScheurenWinkler07}.  Such models output independent  pairwise coreference decisions.  We implemented a mixture model version of the model of \cite{Sadinle14}  as presented in Section \ref{ss:ElSalvador_linkage}.  This approach classifies the record pairs in $\mathcal{C}$ into coreferent and non-coreferent pairs.  The mixture model is obtained by ignoring that the match status of a record pair is given by $M_{ij}=I(Z_i=Z_j)$, and simply taking $M_{ij}\mid p\overset{iid}{\sim} \text{Bernoulli}(p)$, $i< j$.  We used Bayesian estimation of this mixture model employing the same priors for the $m_{fl}$ and $u_{fl}$ parameters as in the  application to the Salvadoran lists, and $p\sim \text{Uniform}(0,1)$.  From running a Gibbs sampler for 10,000 iterations, we obtained a posterior sample of $\{M_{ij}\}_{(i,j)\in \mathcal{C}}$.  

To obtain groups of coreferent records we used transitive closure.  In the mixture model explained above, for each iteration of the Gibbs sampler we obtain a draw of $\{M_{ij}\}_{(i,j)\in \mathcal{C}}$.  For each of these iterations we apply transitive closure by setting $M_{jj'}=1$ if $M_{ij}=M_{ij'}=1$ for any record $i$.  The number of non-transitive triplets $(i,j,j')$, where only two of $M_{ij}$, $M_{ij'}$, $M_{jj'}$ equal to 1, varies between 84 and 156 across the Gibbs iterations, which is not surprising given that this model treats the $M_{ij}$'s as independent.  
Using transitive closure we obtain an ad-hoc constructed distribution of partitions of the records which we can use to implement an ad-hoc version of the linkage-averaged estimate of $N$.  We then proceeded to compute a linkage-averaged posterior using the models of \cite{MadiganYork97} and \cite{Manrique16}, which lead to posterior means of 15,636 and 14,999, and 99\% credible intervals of [6389, 25532] and [5806,  35326], respectively.  In this case the ad-hoc mixture model for record linkage leads to similar results as those obtained using the method of \cite{Sadinle14}.  This can be explained from the fact that both models are essentially the same, with the exception that one samples pairwise matching statuses and the other samples coreference partitions.  Also, the graph induced by the set of candidate coreferent pairs $\mathcal{C}$ is quite sparse and broken into many small connected components, which constrains the clustering effect of transitive closure. Transitive closure can only group records in the same connected component obtained from $\mathcal{C}$.

\subsection{Discussion}

We presented linkage-averaged estimates under individual graphical models, and linkage-averaged two-sample estimates under independence of the list inclusion indicators.  These approaches lead to widely different estimates, but we simply presented them to illustrate the possibilities of linkage-averaging.   The linkage-averaged estimates obtained under the models of \cite{MadiganYork97} and \cite{Manrique16} are more plausible, as they each take into account the uncertainty on the correct model for population size estimation.  

While the same linkage results, in the form of posterior draws of the coreference partition, were used for obtaining all linkage-averaged estimates, the percentage contribution of the linkage uncertainty on the overall uncertainty of $N$ varies with the capture-recapture model.  For some of these approaches the contribution from the linkage is rather small, but we can only measure this after we have computed the linkage-averaged estimates.  

The linkage-averaged posteriors using the models of \cite{MadiganYork97} and \cite{Manrique16} lead to roughly the same point estimates: 13,432 and 13,924 civilian killings, respectively, in the region of San Salvador during the Salvadoran civil war.  The linkage-averaged posteriors themselves, however, disagree in the tails.  The disagreement on the right tail can be explained to some extent when we consider that the prior for $N$ used with the approach of \cite{MadiganYork97} was truncated at 30,000, whereas we did not use this truncation with the approach of \cite{Manrique16} as the implementation of the R package \texttt{LCMCR} does not allow it.  The results using \cite{MadiganYork97} can therefore be seen as somewhat conservative.  

\section{Conclusions}

We presented a linkage-averaging approach to incorporate linkage uncertainty into models for population size estimation.  We used Bayesian partitioning approaches for record linkage which provide posterior distributions on the coreference partition of the records coming from all the data sources.  The models for population size estimation covered by our approach are those whose sufficient statistics are functions of the coreference partition alone.  Under these conditions, linkage-averaging is proper in the sense that it can be derived from a proper Bayesian analysis that combines the record linkage and population size estimation models.  It is important to note, however, that the success of this approach is determined by the success of its components.  For example, if the record linkage model over-links or under-links, then the population size estimates will be lower or higher, respectively, with respect to what we would obtain under the correct linkage.  Similarly, if the model for population size estimation is wrong, our estimates will be deficient regardless of the amount of uncertainty from the linkage stage.  

The class of capture-recapture models considered here is somewhat restrictive given that, for example, they do not allow us to incorporate information on covariates that may influence capture probabilities.  Traditionally, a simple way of dealing with heterogeneous inclusion probabilities in multiple-systems estimation is to stratify by characteristics that influence the inclusion probabilities, such as space and/or time.  To use linkage-averaging to produce population size estimates per stratum (say, year $\times$ region) we would have to assume that the stratifying variables are recorded without error, which might be unreasonable in the context of the datafiles from El Salvador.  For example, suppose two records that disagree in the stratum where they belong are coreferent according to a coreference partition.  Our current methodology does not offer a way of allocating this individual to a unique stratum, nor a way to deal with the uncertainty on where it should be allocated.  However, if the stratifying variables can also be used as blocking variables in the linkage step, then the linkage-averaging approach enjoys Bayesian propriety within each stratum.  In this sense, approaches such as those of \cite{Steortsetal13}, \cite{TancrediLiseo11}, and \cite{LiseoTancredi11} that directly model the information in the datafiles seem promising, given that they explicitly allow us to estimate the latent true values of the individuals in the files.

We also presented an application to the combination of three lists on civilian killings from the civil war of El Salvador.  In this case, the intrinsic variability of Bayesian population size estimation is much larger than the uncertainty coming from the linkage stage, but this might be different in other applications.  Our analyses of the lists from El Salvador indicate that the number of civilian killings during the Salvadoran civil war in the region of San Salvador is most likely to be around 13,000--14,000, but the variability in these estimates is quite large, leading to a posterior 99\% credible interval of [4922, 31429], according to the linkage-averaged estimates obtained using the methodology of \cite{Manrique16}.  Unfortunately, we do not have a way of validating these results, as there does not even exist ground truth for validating the linkage of these datafiles.

\section*{Acknowledgements}
This research is derived from the Ph.D. thesis of the author, supervised by Stephen E. Fienberg.  Steve's many interests included record linkage, population size estimation, and their application to human rights.  The author therefore dedicates this article to the memory of Steve; without his support this research would not have been possible.  

The author also thanks Patrick Ball and Megan Price from the Human Rights Data Analysis Group -- HRDAG for providing access to the data used in this article, and Daniel Manrique-Vallier, Kristian Lum, Robin Mejia, Trivellore Raghunathan, and Thomas Brendan Murphy  for helpful comments that contributed to improving the quality of this article.


\bibliographystyle{imsart-nameyear}
\bibliography{biblio}

\begin{thebibliography}{43}

\bibitem[\protect\citeauthoryear{Anderson and
  Fienberg}{1999}]{AndersonFienbergbook}
\begin{bbook}[author]
\bauthor{\bsnm{Anderson},~\bfnm{Margo~J.}\binits{M.~J.}} \AND
  \bauthor{\bsnm{Fienberg},~\bfnm{Stephen~E.}\binits{S.~E.}}
(\byear{1999}).
\btitle{{Who Counts?: The Politics of Census-Taking in Contemporary America}},
\bedition{Revised paperback (2001)} ed.
\bpublisher{Russell Sage Foundation}, \baddress{New York}.
\end{bbook}
\endbibitem

\bibitem[\protect\citeauthoryear{Ball}{2000}]{Ball00ElSalvador}
\begin{bincollection}[author]
\bauthor{\bsnm{Ball},~\bfnm{Patrick}\binits{P.}}
(\byear{2000}).
\btitle{{The Salvadoran Human Rights Commission: Data Processing, Data
  Representation, and Generating Analytical Reports}}.
In \bbooktitle{{Making the Case: Investigating Large Scale Human Rights
  Violations Using Information Systems and Data Analysis}}
(\beditor{\bfnm{Patrick}\binits{P.}~\bsnm{Ball}},
  \beditor{\bfnm{Herbert~F.}\binits{H.~F.}~\bsnm{Spirer}} \AND
  \beditor{\bfnm{Louise}\binits{L.}~\bsnm{Spirer}}, eds.)
\bpublisher{American Association for the Advancement of Science}.
\end{bincollection}
\endbibitem

\bibitem[\protect\citeauthoryear{Bilenko et~al.}{2003}]{Bilenkoetal03}
\begin{barticle}[author]
\bauthor{\bsnm{Bilenko},~\bfnm{M.}\binits{M.}},
  \bauthor{\bsnm{Mooney},~\bfnm{R.~J.}\binits{R.~J.}},
  \bauthor{\bsnm{Cohen},~\bfnm{W.~W.}\binits{W.~W.}},
  \bauthor{\bsnm{Ravikumar},~\bfnm{P.}\binits{P.}} \AND
  \bauthor{\bsnm{Fienberg},~\bfnm{S.~E.}\binits{S.~E.}}
(\byear{2003}).
\btitle{{Adaptive Name Matching in Information Integration}}.
\bjournal{IEEE Intelligent Systems}
\bvolume{18}
\bpages{16--23}.
\end{barticle}
\endbibitem

\bibitem[\protect\citeauthoryear{Bird and King}{2018}]{BirdKing18}
\begin{barticle}[author]
\bauthor{\bsnm{Bird},~\bfnm{Sheila~M.}\binits{S.~M.}} \AND
  \bauthor{\bsnm{King},~\bfnm{Ruth}\binits{R.}}
(\byear{2018}).
\btitle{{Multiple Systems Estimation (or Capture-ReCapture Estimation) to
  Inform Public Policy}}.
\bjournal{Annual Review of Statistics and Its Application}
\bvolume{5}.
\end{barticle}
\endbibitem

\bibitem[\protect\citeauthoryear{Bishop, Fienberg and
  Holland}{1975}]{Bishopetal75}
\begin{bbook}[author]
\bauthor{\bsnm{Bishop},~\bfnm{Yvonne~M.}\binits{Y.~M.}},
  \bauthor{\bsnm{Fienberg},~\bfnm{Stephen~E.}\binits{S.~E.}} \AND
  \bauthor{\bsnm{Holland},~\bfnm{Paul~W.}\binits{P.~W.}}
(\byear{1975}).
\btitle{{Discrete Multivariate Analysis: Theory and Practice}}.
\bpublisher{The MIT Press. Reprinted in 2007 by Springer, New York}.
\end{bbook}
\endbibitem

\bibitem[\protect\citeauthoryear{Castledine}{1981}]{Castledine81}
\begin{barticle}[author]
\bauthor{\bsnm{Castledine},~\bfnm{B.~J.}\binits{B.~J.}}
(\byear{1981}).
\btitle{{A Bayesian Analysis of Multiple-Recapture Sampling for a Closed
  Population}}.
\bjournal{Biometrika}
\bvolume{68}
\bpages{197--210}.
\end{barticle}
\endbibitem

\bibitem[\protect\citeauthoryear{Christen}{2012}]{Christen12}
\begin{barticle}[author]
\bauthor{\bsnm{Christen},~\bfnm{Peter}\binits{P.}}
(\byear{2012}).
\btitle{{A Survey of Indexing Techniques for Scalable Record Linkage and
  Deduplication}}.
\bjournal{IEEE Transactions on Knowledge and Data Engineering}
\bvolume{24}
\bpages{1537--1555}.
\end{barticle}
\endbibitem

\bibitem[\protect\citeauthoryear{Dawid and Lauritzen}{1993}]{DawidLauritzen93}
\begin{barticle}[author]
\bauthor{\bsnm{Dawid},~\bfnm{A.~P.}\binits{A.~P.}} \AND
  \bauthor{\bsnm{Lauritzen},~\bfnm{S.~L.}\binits{S.~L.}}
(\byear{1993}).
\btitle{{Hyper Markov Laws in the Statistical Analysis of Decomposable
  Graphical Models}}.
\bjournal{Annals of Statistics}
\bvolume{21}
\bpages{1272--1317}.
\end{barticle}
\endbibitem

\bibitem[\protect\citeauthoryear{Edwards}{2000}]{Edwards00}
\begin{bbook}[author]
\bauthor{\bsnm{Edwards},~\bfnm{David}\binits{D.}}
(\byear{2000}).
\btitle{{Introduction to Graphical Modelling}},
\bedition{2nd} ed.
\bpublisher{Springer-Verlag}.
\end{bbook}
\endbibitem

\bibitem[\protect\citeauthoryear{Elmagarmid, Ipeirotis and
  Verykios}{2007}]{Elmagarmidetal07}
\begin{barticle}[author]
\bauthor{\bsnm{Elmagarmid},~\bfnm{Ahmed~K.}\binits{A.~K.}},
  \bauthor{\bsnm{Ipeirotis},~\bfnm{Panagiotis~G.}\binits{P.~G.}} \AND
  \bauthor{\bsnm{Verykios},~\bfnm{Vassilios~S.}\binits{V.~S.}}
(\byear{2007}).
\btitle{{Duplicate Record Detection: A Survey}}.
\bjournal{IEEE Transactions on Knowledge and Data Engineering}
\bvolume{19}
\bpages{1--16}.
\end{barticle}
\endbibitem

\bibitem[\protect\citeauthoryear{Ericksen, Kadane and
  Tukey}{1989}]{EricksenKadaneTukey89}
\begin{barticle}[author]
\bauthor{\bsnm{Ericksen},~\bfnm{Eugene~P.}\binits{E.~P.}},
  \bauthor{\bsnm{Kadane},~\bfnm{Joseph~B.}\binits{J.~B.}} \AND
  \bauthor{\bsnm{Tukey},~\bfnm{John~W.}\binits{J.~W.}}
(\byear{1989}).
\btitle{{Adjusting the 1980 Census of Population and Housing}}.
\bjournal{Journal of the American Statistical Association}
\bvolume{84}
\bpages{927--944}.
\end{barticle}
\endbibitem

\bibitem[\protect\citeauthoryear{Fellegi and Sunter}{1969}]{FellegiSunter69}
\begin{barticle}[author]
\bauthor{\bsnm{Fellegi},~\bfnm{Ivan~P.}\binits{I.~P.}} \AND
  \bauthor{\bsnm{Sunter},~\bfnm{Alan~B.}\binits{A.~B.}}
(\byear{1969}).
\btitle{{A Theory for Record Linkage}}.
\bjournal{Journal of the American Statistical Association}
\bvolume{64}
\bpages{1183--1210}.
\end{barticle}
\endbibitem

\bibitem[\protect\citeauthoryear{Fienberg}{1972}]{Fienberg72}
\begin{barticle}[author]
\bauthor{\bsnm{Fienberg},~\bfnm{Stephen~E.}\binits{S.~E.}}
(\byear{1972}).
\btitle{{The Multiple Recapture Census for Closed Populations and Incomplete
  {2$^{k}$} Contingency Tables}}.
\bjournal{Biometrika}
\bvolume{59}
\bpages{591-603}.
\end{barticle}
\endbibitem

\bibitem[\protect\citeauthoryear{Fienberg, Johnson and
  Junker}{1999}]{FienbergJohnsonJunker99}
\begin{barticle}[author]
\bauthor{\bsnm{Fienberg},~\bfnm{Stephen~E.}\binits{S.~E.}},
  \bauthor{\bsnm{Johnson},~\bfnm{Matthew~S.}\binits{M.~S.}} \AND
  \bauthor{\bsnm{Junker},~\bfnm{Brian~W.}\binits{B.~W.}}
(\byear{1999}).
\btitle{{Classical Multilevel and Bayesian Approaches to Population Size
  Estimation Using Multiple Lists}}.
\bjournal{Journal of the Royal Statistical Society. Series A (Statistics in
  Society)}
\bvolume{162}
\bpages{383--405}.
\end{barticle}
\endbibitem

\bibitem[\protect\citeauthoryear{Fortini et~al.}{2002}]{Fortinietal02}
\begin{binproceedings}[author]
\bauthor{\bsnm{Fortini},~\bfnm{M.}\binits{M.}},
  \bauthor{\bsnm{Nuccitelli},~\bfnm{A.}\binits{A.}},
  \bauthor{\bsnm{Liseo},~\bfnm{B.}\binits{B.}} \AND
  \bauthor{\bsnm{Scanu},~\bfnm{M}\binits{M.}}
(\byear{2002}).
\btitle{{Modeling Issues in Record Linkage: A Bayesian Perspective}}.
In \bbooktitle{Proceedings of the Section on Survey Research Methods}
\bpages{1008--1013}.
\bpublisher{American Statistical Association}.
\end{binproceedings}
\endbibitem

\bibitem[\protect\citeauthoryear{George and Robert}{1992}]{GeorgeRobert92}
\begin{barticle}[author]
\bauthor{\bsnm{George},~\bfnm{Edward~I.}\binits{E.~I.}} \AND
  \bauthor{\bsnm{Robert},~\bfnm{Christian~P.}\binits{C.~P.}}
(\byear{1992}).
\btitle{{Capture-Recapture Estimation Via Gibbs Sampling}}.
\bjournal{Biometrika}
\bvolume{79}
\bpages{677--683}.
\end{barticle}
\endbibitem

\bibitem[\protect\citeauthoryear{Gutman, Afendulis and
  Zaslavsky}{2013}]{Gutmanetal13}
\begin{barticle}[author]
\bauthor{\bsnm{Gutman},~\bfnm{Roee}\binits{R.}},
  \bauthor{\bsnm{Afendulis},~\bfnm{Christopher~C.}\binits{C.~C.}} \AND
  \bauthor{\bsnm{Zaslavsky},~\bfnm{Alan~M.}\binits{A.~M.}}
(\byear{2013}).
\btitle{{A Bayesian Procedure for File Linking to Analyze End-of-Life Medical
  Costs}}.
\bjournal{Journal of the American Statistical Association}
\bvolume{108}
\bpages{34--47}.
\end{barticle}
\endbibitem

\bibitem[\protect\citeauthoryear{Herzog, Scheuren and
  Winkler}{2007}]{HerzogScheurenWinkler07}
\begin{bbook}[author]
\bauthor{\bsnm{Herzog},~\bfnm{Thomas~N.}\binits{T.~N.}},
  \bauthor{\bsnm{Scheuren},~\bfnm{Fritz~J.}\binits{F.~J.}} \AND
  \bauthor{\bsnm{Winkler},~\bfnm{William~E.}\binits{W.~E.}}
(\byear{2007}).
\btitle{{Data Quality and Record Linkage Techniques}}.
\bpublisher{Springer}, \baddress{New York}.
\end{bbook}
\endbibitem

\bibitem[\protect\citeauthoryear{Hogan}{1992}]{Hogan92}
\begin{barticle}[author]
\bauthor{\bsnm{Hogan},~\bfnm{Howard}\binits{H.}}
(\byear{1992}).
\btitle{{The 1990 Post-Enumeration Survey: An Overview}}.
\bjournal{The American Statistician}
\bvolume{46}
\bpages{261--269}.
\end{barticle}
\endbibitem

\bibitem[\protect\citeauthoryear{Hogan}{1993}]{Hogan93}
\begin{barticle}[author]
\bauthor{\bsnm{Hogan},~\bfnm{Howard}\binits{H.}}
(\byear{1993}).
\btitle{{The 1990 Post-Enumeration Survey: Operations and Results}}.
\bjournal{Journal of the American Statistical Association}
\bvolume{88}
\bpages{1047--1060}.
\end{barticle}
\endbibitem

\bibitem[\protect\citeauthoryear{Howland}{2008}]{Howland08}
\begin{barticle}[author]
\bauthor{\bsnm{Howland},~\bfnm{Todd}\binits{T.}}
(\byear{2008}).
\btitle{{How El Rescate, a Small Nongovernmental Organization, Contributed to
  the Transformation of the Human Rights Situation in El Salvador}}.
\bjournal{Human Rights Quarterly}
\bvolume{30}
\bpages{703--757}.
\end{barticle}
\endbibitem

\bibitem[\protect\citeauthoryear{Ishwaran and James}{2001}]{IshwaranJames01}
\begin{barticle}[author]
\bauthor{\bsnm{Ishwaran},~\bfnm{H.}\binits{H.}} \AND
  \bauthor{\bsnm{James},~\bfnm{L.~F.}\binits{L.~F.}}
(\byear{2001}).
\btitle{{Gibbs sampling for stick-breaking priors}}.
\bjournal{Journal of the American Statistical Association}
\bvolume{96}
\bpages{161--173}.
\end{barticle}
\endbibitem

\bibitem[\protect\citeauthoryear{Jaro}{1989}]{Jaro89}
\begin{barticle}[author]
\bauthor{\bsnm{Jaro},~\bfnm{Matthew~A.}\binits{M.~A.}}
(\byear{1989}).
\btitle{{Advances in Record-Linkage Methodology as Applied to Matching the 1985
  Census of Tampa, Florida}}.
\bjournal{Journal of the American Statistical Association}
\bvolume{84}
\bpages{414--420}.
\end{barticle}
\endbibitem

\bibitem[\protect\citeauthoryear{LaPorte et~al.}{1993}]{LaPorteetal93}
\begin{barticle}[author]
\bauthor{\bsnm{LaPorte},~\bfnm{Ronald~E}\binits{R.~E.}},
  \bauthor{\bsnm{McCarty},~\bfnm{Daniel}\binits{D.}},
  \bauthor{\bsnm{Bruno},~\bfnm{Graziella}\binits{G.}},
  \bauthor{\bsnm{Tajima},~\bfnm{Naoko}\binits{N.}} \AND
  \bauthor{\bsnm{Baba},~\bfnm{Shigeaki}\binits{S.}}
(\byear{1993}).
\btitle{Counting Diabetes in the Next Millennium: Application of
  capture-recapture technology}.
\bjournal{Diabetes Care}
\bvolume{16}
\bpages{528--534}.
\end{barticle}
\endbibitem

\bibitem[\protect\citeauthoryear{Larsen and Rubin}{2001}]{LarsenRubin01}
\begin{barticle}[author]
\bauthor{\bsnm{Larsen},~\bfnm{Michael~D.}\binits{M.~D.}} \AND
  \bauthor{\bsnm{Rubin},~\bfnm{Donald~B.}\binits{D.~B.}}
(\byear{2001}).
\btitle{{Iterative Automated Record Linkage Using Mixture Models}}.
\bjournal{Journal of the American Statistical Association}
\bvolume{96}
\bpages{32--41}.
\end{barticle}
\endbibitem

\bibitem[\protect\citeauthoryear{Lauritzen}{1996}]{Lauritzen96}
\begin{bbook}[author]
\bauthor{\bsnm{Lauritzen},~\bfnm{Steffen}\binits{S.}}
(\byear{1996}).
\btitle{{Graphical Models}}.
\bpublisher{Oxford University Press}.
\end{bbook}
\endbibitem

\bibitem[\protect\citeauthoryear{Liseo and Tancredi}{2011}]{LiseoTancredi11}
\begin{barticle}[author]
\bauthor{\bsnm{Liseo},~\bfnm{Brunero}\binits{B.}} \AND
  \bauthor{\bsnm{Tancredi},~\bfnm{Andrea}\binits{A.}}
(\byear{2011}).
\btitle{{Bayesian Estimation of Population Size via Linkage of Multivariate
  Normal Data Sets}}.
\bjournal{Journal of Official Statistics}
\bvolume{27}
\bpages{491--505}.
\end{barticle}
\endbibitem

\bibitem[\protect\citeauthoryear{Lum, Price and Banks}{2013}]{LumPriceBanks13}
\begin{barticle}[author]
\bauthor{\bsnm{Lum},~\bfnm{Kristian}\binits{K.}},
  \bauthor{\bsnm{Price},~\bfnm{Megan~Emily}\binits{M.~E.}} \AND
  \bauthor{\bsnm{Banks},~\bfnm{David}\binits{D.}}
(\byear{2013}).
\btitle{{Applications of Multiple Systems Estimation in Human Rights
  Research}}.
\bjournal{The American Statistician}
\bvolume{67}
\bpages{191--200}.
\end{barticle}
\endbibitem

\bibitem[\protect\citeauthoryear{Madigan and York}{1997}]{MadiganYork97}
\begin{barticle}[author]
\bauthor{\bsnm{Madigan},~\bfnm{David}\binits{D.}} \AND
  \bauthor{\bsnm{York},~\bfnm{Jeremy~C.}\binits{J.~C.}}
(\byear{1997}).
\btitle{{Bayesian Methods for Estimation of the Size of a Closed Population}}.
\bjournal{Biometrika}
\bvolume{1}
\bpages{19--31}.
\end{barticle}
\endbibitem

\bibitem[\protect\citeauthoryear{Manrique-Vallier}{2016}]{Manrique16}
\begin{barticle}[author]
\bauthor{\bsnm{Manrique-Vallier},~\bfnm{Daniel}\binits{D.}}
(\byear{2016}).
\btitle{{Bayesian Population Size Estimation Using Dirichlet Process
  Mixtures}}.
\bjournal{Biometrics}
\bvolume{72}
\bpages{1246--1254}.
\end{barticle}
\endbibitem

\bibitem[\protect\citeauthoryear{Matsakis}{2010}]{Matsakis10}
\begin{bphdthesis}[author]
\bauthor{\bsnm{Matsakis},~\bfnm{Nicholas~Elias}\binits{N.~E.}}
(\byear{2010}).
\btitle{{Active Duplicate Detection with Bayesian Nonparametric Models}}
\btype{PhD thesis},
\bpublisher{Massachusetts Institute of Technology}.
\end{bphdthesis}
\endbibitem

\bibitem[\protect\citeauthoryear{{Commission on the Truth for El
  Salvador}}{1993}]{CTElSalvador93}
\begin{bmisc}[author]
\bauthor{\bsnm{{Commission on the Truth for El Salvador}}}
(\byear{1993}).
\btitle{{From Madness to Hope: the 12-Year War in El Salvador: Report of the
  Commission on the Truth for El Salvador}}.
\bhowpublished{http://www.usip.org/files/file/ElSalvador-Report.pdf [Accessed
  May 31, 2013]}.
\bnote{UN Security Council}.
\end{bmisc}
\endbibitem

\bibitem[\protect\citeauthoryear{Plummer et~al.}{2006}]{CODA}
\begin{barticle}[author]
\bauthor{\bsnm{Plummer},~\bfnm{Martyn}\binits{M.}},
  \bauthor{\bsnm{Best},~\bfnm{Nicky}\binits{N.}},
  \bauthor{\bsnm{Cowles},~\bfnm{Kate}\binits{K.}} \AND
  \bauthor{\bsnm{Vines},~\bfnm{Karen}\binits{K.}}
(\byear{2006}).
\btitle{{CODA: Convergence Diagnosis and Output Analysis for MCMC}}.
\bjournal{R News}
\bvolume{6}
\bpages{7--11}.
\end{barticle}
\endbibitem

\bibitem[\protect\citeauthoryear{Pollock}{2000}]{Pollock00}
\begin{barticle}[author]
\bauthor{\bsnm{Pollock},~\bfnm{K.~H.}\binits{K.~H.}}
(\byear{2000}).
\btitle{{Capture-Recapture Models}}.
\bjournal{Journal of the American Statistical Association}
\bvolume{95}
\bpages{293--296}.
\end{barticle}
\endbibitem

\bibitem[\protect\citeauthoryear{Price and Ball}{2015}]{PriceBall15}
\begin{barticle}[author]
\bauthor{\bsnm{Price},~\bfnm{Megan}\binits{M.}} \AND
  \bauthor{\bsnm{Ball},~\bfnm{Patrick}\binits{P.}}
(\byear{2015}).
\btitle{{Selection Bias and the Statistical Patterns of Mortality in
  Conflict}}.
\bjournal{Statistical Journal of the IAOS}
\bvolume{31}
\bpages{263--272}.
\end{barticle}
\endbibitem

\bibitem[\protect\citeauthoryear{Price, Gohdes and Ball}{2015}]{Priceetal15}
\begin{barticle}[author]
\bauthor{\bsnm{Price},~\bfnm{Megan}\binits{M.}},
  \bauthor{\bsnm{Gohdes},~\bfnm{Anita}\binits{A.}} \AND
  \bauthor{\bsnm{Ball},~\bfnm{Patrick}\binits{P.}}
(\byear{2015}).
\btitle{{Documents of War: Understanding the Syrian Conflict}}.
\bjournal{Significance}
\bvolume{12}
\bpages{14--19}.
\end{barticle}
\endbibitem

\bibitem[\protect\citeauthoryear{Sadinle}{2014}]{Sadinle14}
\begin{barticle}[author]
\bauthor{\bsnm{Sadinle},~\bfnm{Mauricio}\binits{M.}}
(\byear{2014}).
\btitle{{Detecting Duplicates in a Homicide Registry Using a Bayesian
  Partitioning Approach}}.
\bjournal{Annals of Applied Statistics}
\bvolume{8}
\bpages{2404--2434}.
\end{barticle}
\endbibitem

\bibitem[\protect\citeauthoryear{Sadinle}{2017}]{Sadinle17}
\begin{barticle}[author]
\bauthor{\bsnm{Sadinle},~\bfnm{Mauricio}\binits{M.}}
(\byear{2017}).
\btitle{{Bayesian Estimation of Bipartite Matchings for Record Linkage}}.
\bjournal{Journal of the American Statistical Association}
\bvolume{112}
\bpages{600-612}.
\end{barticle}
\endbibitem

\bibitem[\protect\citeauthoryear{Steorts}{2015}]{Steorts15}
\begin{barticle}[author]
\bauthor{\bsnm{Steorts},~\bfnm{Rebecca~C.}\binits{R.~C.}}
(\byear{2015}).
\btitle{{Entity Resolution with Empirically Motivated Priors}}.
\bjournal{Bayesian Analysis}
\bvolume{10}
\bpages{849--875}.
\end{barticle}
\endbibitem

\bibitem[\protect\citeauthoryear{Steorts, Hall and
  Fienberg}{2016}]{Steortsetal13}
\begin{barticle}[author]
\bauthor{\bsnm{Steorts},~\bfnm{Rebecca~C.}\binits{R.~C.}},
  \bauthor{\bsnm{Hall},~\bfnm{Rob}\binits{R.}} \AND
  \bauthor{\bsnm{Fienberg},~\bfnm{Stephen~E.}\binits{S.~E.}}
(\byear{2016}).
\btitle{{A Bayesian Approach to Graphical Record Linkage and Deduplication}}.
\bjournal{Journal of the American Statistical Association}
\bvolume{111}
\bpages{1660--1672}.
\end{barticle}
\endbibitem

\bibitem[\protect\citeauthoryear{Tancredi and Liseo}{2011}]{TancrediLiseo11}
\begin{barticle}[author]
\bauthor{\bsnm{Tancredi},~\bfnm{Andrea}\binits{A.}} \AND
  \bauthor{\bsnm{Liseo},~\bfnm{Brunero}\binits{B.}}
(\byear{2011}).
\btitle{{A Hierarchical Bayesian Approach to Record Linkage and Size Population
  Problems}}.
\bjournal{Annals of Applied Statistics}
\bvolume{5}
\bpages{1553--1585}.
\end{barticle}
\endbibitem

\bibitem[\protect\citeauthoryear{Winkler}{1988}]{Winkler88}
\begin{binproceedings}[author]
\bauthor{\bsnm{Winkler},~\bfnm{W.~E.}\binits{W.~E.}}
(\byear{1988}).
\btitle{{Using the EM Algorithm for Weight Computation in the Fellegi-Sunter
  Model of Record Linkage}}.
In \bbooktitle{Proceedings of the Section on Survey Research Methods}
\bpages{667--671}.
\bpublisher{American Statistical Association}.
\end{binproceedings}
\endbibitem

\bibitem[\protect\citeauthoryear{Winkler}{1990}]{Winkler90Strings}
\begin{binproceedings}[author]
\bauthor{\bsnm{Winkler},~\bfnm{W.~E.}\binits{W.~E.}}
(\byear{1990}).
\btitle{{String Comparator Metrics and Enhanced Decision Rules in the
  Fellegi-Sunter Model of Record Linkage}}.
In \bbooktitle{Proceedings of the Section on Survey Research Methods}
\bpages{354--359}.
\bpublisher{American Statistical Association}.
\end{binproceedings}
\endbibitem

\end{thebibliography}

\end{document}